\documentclass[letterpaper]{article}
\usepackage[preprint]{aaai2027}
\usepackage[hyphens]{url}
\usepackage{graphicx}
\usepackage{natbib}
\usepackage{caption}
\usepackage{booktabs}
\usepackage{tabularx}
\usepackage{amsmath,amssymb,amsthm}
\usepackage{mathtools}
\usepackage{algorithm}
\usepackage{algorithmic}
\usepackage{placeins}

\newtheorem{theorem}{Theorem}
\newtheorem{proposition}{Proposition}
\newtheorem{lemma}{Lemma}
\newtheorem{corollary}{Corollary}

\DeclareCaptionStyle{ruled}{labelfont=normalfont,labelsep=colon,strut=off}

\title{ATLAS: Learning to Recommend Across Unseen Domains}

\author{
    Pervez Shaik,
    Prosenjit Biswas,
    Abhinav Thorat,
    Ravi Kolla,
    Niranjan Pedanekar
}
\affiliations{
    User Engagement Research, Sony Research India\\
    \{shaik.pervez, prosenjit.biswas, abhinav.thorat, ravi.kolla, niranjan.pedanekar\}@sony.com
}

\begin{document}

\maketitle

\begin{abstract}

Recommender systems remain fundamentally domain-bound: models trained on one interaction environment typically require retraining or target-domain adaptation before they can operate on a new catalogue. For example, a recommender trained on movies cannot be directly deployed to recommend groceries or video games. Existing approaches mitigate this limitation by transferring restricted forms of recommendation knowledge, adapting to the target domain, or leveraging large language models (LLMs) to obtain transferable representations. We instead investigate whether recommendation-specific knowledge learned solely from multiple heterogeneous domains can generalize to entirely unseen domains without target-domain adaptation or language-model pretraining. To this end, we introduce ATLAS, a multi-source recommendation domain generalization framework that learns a shared, domain-invariant user--item representation from multiple disjoint source domains, enabling zero-shot recommendation on unseen domains. ATLAS builds this shared space through three components: a Gromov--Wasserstein alignment that preserves how users relate to one another across domains, and an adversarial objective that makes item representations indistinguishable across domains. These are combined with residual vector quantization (RVQ) codebooks that compress user and item embeddings into a discrete latent space, capturing hierarchical interaction patterns while suppressing domain-specific variation. The contribution of each component is validated through comprehensive ablation studies and post-hoc analyses of the learned representation space. Through extensive experiments, we demonstrate that when ATLAS is trained on five Amazon domains and directly applied to ten entirely unseen domains, it outperforms state-of-the-art sequential, graph-based, cross-domain, quantization-based, and LLM-based recommendation baselines on the majority of unseen domains, with an average relative gain in HitRate (HR) of 24\%. Furthermore, we identify a pronounced source-domain diversity effect, showing that increasing the heterogeneity of source domains substantially improves zero-shot transfer performance. These findings demonstrate that ATLAS establishes recommendation domain generalization as a promising paradigm for zero-shot recommendation beyond target-domain adaptation and LLM-based recommenders.
\end{abstract}

\section{Introduction}
\label{sec:introduction}
Modern recommender systems remain largely domain-specific. Models trained
on one interaction environment typically assume that deployment users,
items, and behavioural distributions resemble those observed during
training. This assumption limits their reuse across new catalogues,
markets, and product categories, where user and item representations
must be rebuilt and the model retrained. Strong
collaborative, graph-based, and sequential recommenders such as
LightGCN \cite{he2020lightgcn}, SASRec \cite{kang2018self}, and BERT4Rec \cite{sun2019bert4rec} therefore achieve high in-domain accuracy
but do not naturally transfer to unseen domains containing entirely new users
and items.

Cross-domain recommendation seeks to address this limitation by learning domain-invariant signals between related source and target domains. Existing methods exploit overlapping entities, shared latent spaces, distribution alignment, or target-domain fine-tuning, and more recent approaches support partially or fully non-overlapping users, but the target domain generally participates during training, adaptation, or joint optimization; they therefore address transfer to a known target rather than deployment to a domain absent throughout learning. Universal and pretrained recommenders broaden transfer beyond a single source--target pair: UniSRec \cite{hou2022towards} and ZESRec \cite{ding2021zero} learn transferable sequential or semantic representations, PrepRec \cite{wang2024pre} and PreRec \cite{lin2024pre} transfer popularity patterns or recommendation priors across datasets, and foundation-style systems such as P5 \cite{geng2022recommendation}, TALLRec \cite{bao2023tallrec}, and RecBase \cite{zhou2025recbase} instead rely on language-model or large-scale generative pretraining. These methods substantially improve recommendation transfer, but they obtain it through target adaptation, specialised transferable signals, or language and generative pretraining, and do not explicitly study whether one frozen retrieval model, learned solely from multiple heterogeneous source domains, can directly serve multiple completely unseen recommendation domains.

We study this problem under a setting we term \textbf{Recommendation Domain Generalization (RDG)}. In this setting, a model is trained jointly on (K) heterogeneous source domains and subsequently applied to a target domain that is unseen during training, without any parameter updates. The source and target domains share neither users nor items. While target-user histories may be used to construct personalised queries at inference time, no target interaction is ever used for optimisation, model selection, or representation adaptation. RDG thus differs fundamentally from cross-domain recommendation, in which the target domain is seen during training, as well as from universal pretraining followed by downstream fine-tuning. 

We propose ATLAS, a multi-source framework for RDG that operates on the central hypothesis that heterogeneous recommendation domains share transferable user--item interaction structure despite differences in their users, catalogues, and semantics. We refer to this shared structure as \emph{recommendation knowledge}, and posit that if source-specific correlations can be suppressed while this common structure is preserved, a scoring function learned across diverse source domains should remain useful in an unseen domain. Motivated by this hypothesis, ATLAS learns a common user--item retrieval space by aligning complementary aspects of source-domain representations and organising the resulting space with shared hierarchical codebooks. The learned model is frozen after source training and performs similarity-based retrieval directly in unseen domains, without target-domain retraining or language-model inference.

We evaluate ATLAS trained on five Amazon domains and
test it on ten unseen domains with disjoint users and items. We compare
against strong sequential, graph-based, cross-domain, universal, and
LLM-based recommenders under standard protocols. ATLAS
consistently improves zero-shot recommendation performance while also
remaining competitive on the source domains. Our analysis further
reveals a pronounced source-domain diversity effect: expanding training
from fewer to more heterogeneous source domains substantially improves
unseen-domain transfer. These findings suggest that transferable
recommendation behaviour can emerge directly from diverse
recommendation environments rather than requiring target-domain
adaptation or language-model pretraining.

Our contributions are:

\begin{itemize}
    \item We formulate \textbf{Recommendation Domain Generalization}, in which a single model trained on multiple source domains is deployed without parameter updates to unseen domains with disjoint users and items.
    
    \item We propose ATLAS, a multi-source retrieval framework that learns transferable user--item representations through complementary alignment and shared hierarchical codebooks.
    
    \item We conduct extensive evaluation over five source and ten unseen target domains, comparing against domain-specific, cross-domain, universal, and LLM-based recommendation methods.
    
    \item We identify a strong \textbf{source-domain diversity effect}, showing that broader source-domain exposure of ATLAS substantially improves its zero-shot recommendation performance.
\end{itemize}

\section{Related Work}
\label{sec:literature-survey}

\textbf{Domain-Specific Recommendation.}
Domain-specific recommenders such as LightGCN, SASRec, and BERT4Rec learn user--item representations from historical interactions within a fixed recommendation environment. However, they assume consistent users, items, and interaction distributions between training and deployment, requiring retraining for new domains.\\
\textbf{Cross-Domain and Universal Recommendation.}
Cross-domain recommendation transfers knowledge between related recommendation domains. Early methods such as CoNet \cite{hu2018conet}, and EMCDR \cite{man2017cross} rely on overlapping users to transfer knowledge across domains, whereas MVDNN \cite{inproceedings} learns a shared latent representation to bridge domain-specific user preferences, while later approaches employ adversarial learning, distribution alignment, and optimal transport (e.g., GWCDR \cite{li2022gromov}, VDEA \cite{liu2022exploiting}, DUP-OT \cite{xiao2026modeling}) to support transfer with limited or non-overlapping users. Nevertheless, target-domain interactions typically remain part of training or adaptation.
Universal recommendation extends transfer beyond a single source--target pair through multi-domain recommendation pretraining. UniSRec \cite{hou2022towards} and ZESRec \cite{ding2021zero} learn transferable sequential or semantic representations, while PreRec \cite{lin2024pre} investigate transferable recommendation priors across multiple datasets. Unlike RDG, these methods are primarily designed for downstream adaptation rather than direct deployment to completely unseen domains using a frozen retrieval model.\\
\textbf{Foundation Models and LLM-based Recommendation.}
Foundation models and LLMs have recently been applied to recommendation through text generation, instruction tuning, prompting, and retrieval augmentation. Representative examples include P5\cite{geng2022recommendation}, TALLRec\cite{bao2023tallrec}, LLMRec \cite{lyu2024llm}, and RecBase \cite{zhou2025recbase}. Unlike these approaches, ATLAS investigates whether transferable recommendation behavior can emerge directly from recommendation-specific representation learning without language-model inference or target-domain adaptation.\\
\textbf{Domain Generalization and Representation Alignment}
Outside recommendation, domain generalization learns representations that generalize to unseen domains without target-domain data. Existing approaches include domain-adversarial learning, invariant representation learning, optimal transport, Gromov--Wasserstein alignment, and vector quantization. Rather than introducing new alignment techniques, ATLAS adapts these established principles to recommendation under the proposed RDG setting. These differences are summarized in Table~\ref{tab:transfer_paradigms} and provides motivation for our RDG setting.

\begin{table}[t]
\small
\centering
\caption{Comparison of recommendation transfer paradigms.}
\label{tab:transfer_paradigms}
\small
\setlength{\tabcolsep}{2.7pt}
\renewcommand{\arraystretch}{1.15}

\begin{tabularx}{\linewidth}{
>{\raggedright\arraybackslash}p{0.22\linewidth}
>{\raggedright\arraybackslash}p{0.30\linewidth}
>{\centering\arraybackslash}p{0.10\linewidth}
>{\centering\arraybackslash}p{0.14\linewidth}
>{\centering\arraybackslash}p{0.18\linewidth}
}
\toprule
\textbf{Paradigm} 
& \textbf{Representative Works} 
& \textbf{Seen} 
& \textbf{Target Adapt.} 
& \textbf{Frozen Zero-Shot} \\
\midrule

Domain-Specific
& LightGCN, SASRec, BERT4Rec
& No
& N/A
& No \\

Cross-Domain
& CoNet, EMCDR, GWCDR
& Yes
& Yes
& No \\

Universal Recommendation
& UniSRec, ZESRec, PrepRec, PreRec
& Partial
& Usually
& Partial \\

Foundation / LLM
& P5, TALLRec, RecBase
& Partial
& Prompt/FT
& Partial \\

\textbf{RDG}
& \textbf{ATLAS}
& \textbf{No}
& \textbf{No}
& \textbf{Yes} \\

\bottomrule
\end{tabularx}
\end{table}

\section{Proposed Model}
\label{sec:proposed-model}
\subsection{Problem Formulation}
Let $\mathcal{D}_{\mathrm{src}}=\{D_1,\ldots,D_K\}$ denote a collection of $K$ heterogeneous source recommendation domains, where each domain $D_k=(U_k,I_k,\mathcal{R}_k)$ consists of a user set $U_k$, an item set $I_k$, and observed user--item interactions $\mathcal{R}_k$. Throughout this work, the users and items are assumed to be disjoint across domains. Our objective is to learn a single recommendation model from $\mathcal{D}_{\mathrm{src}}$ that generalizes directly to any unseen target domain $D_t=(U_t,I_t,\mathcal{R}_t)$, where $D_t\notin\mathcal{D}_{\mathrm{src}}$. Unlike cross-domain recommendation, which assumes the target domain is available during training or adaptation, RDG requires the learned model to remain completely frozen after training on the source domains. At inference, target-user interaction histories are used only to construct user representations for retrieval; they are never used to update model parameters.
\begin{figure*}[t]
\centering
\includegraphics[width=\textwidth]{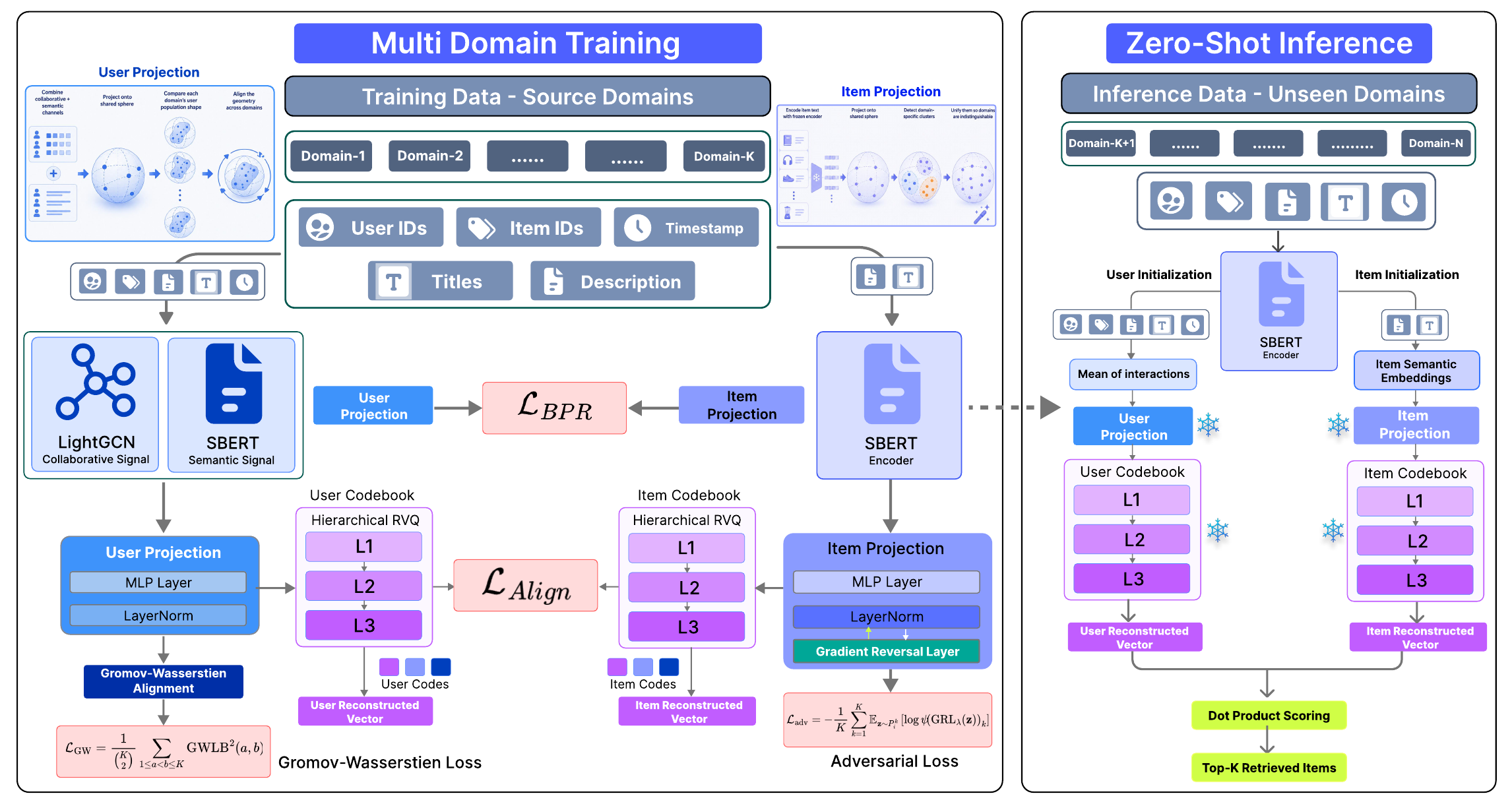}
\caption{ATLAS architecture}
\label{fig:proposed-model-architecture}
\end{figure*}
\subsection{ATLAS Framework}
To address the proposed RDG setting, we introduce ATLAS, a multi-source recommendation framework
that learns a unified retrieval space from multiple heterogeneous source
domains. Rather than relying on target-domain adaptation, ATLAS jointly
aligns item semantics and user interaction geometry into a common latent
space, which is subsequently organized using shared hierarchical
codebooks. The resulting retrieval model is trained exclusively on
source-domain interactions and remains completely frozen during
deployment to unseen recommendation domains. Figure~\ref{fig:proposed-model-architecture} illustrates the overall architecture of ATLAS, while Algorithm~\ref{alg:training} summarizes the training procedure.\\
ATLAS transforms these heterogeneous representations into a shared latent
space through three successive stages:
(i) semantic alignment of item representations,
(ii) geometric alignment of user representations, and
(iii) learning shared hierarchical codebooks for unified retrieval.
The following sections describe each component in detail.

\subsection{Item-Space Unification}
\textbf{Item Representation.} Each item is represented by the concatenation of its title and description and encoded using a frozen Sentence-BERT encoder
$E_{\mathrm{text}}$. Because the encoder is shared across all domains and remains fixed during training, every item is embedded into a common semantic space, $e_i^k = E_{\mathrm{text}}(\text{title},\text{description}).$ \\
\textbf{Shared Item Projection.}
Semantic embeddings are mapped into the recommendation space through a
shared projection network $z_i^k = P_i(e_i^k) / \|P_i(e_i^k)\|_2, $
where $P_i$ is a trainable MLP shared across all source domains.
The $\ell_2$ normalization projects all embeddings onto the unit
hypersphere, ensuring cosine similarity is comparable across domains and
providing a common geometry for subsequent retrieval. \\
\textbf{Adversarial Unification.} A shared space is not yet a \emph{unified} one. Differences in vocabulary, description style, and topical composition make the per-domain item distributions $P^{1}_{i},\dots,P^{K}_{i}$ occupy distinct clusters of $\mathbb{S}^{d-1}$, and an unseen domain then lands in its own cluster, where the scorer has received no training signal. We collapse these clusters by making the domain of an item statistically unrecoverable from its embedding\cite{ben2010theory}. A discriminator $\psi:\mathbb{S}^{d-1}\!\to\!\Delta^{K-1}$ predicts the domain, and $P_i$ is trained against it:
\begin{equation}
\min_{\theta_{\psi}}\ \max_{\theta_{P_i}}\
-\frac{1}{K}\sum\nolimits_{k=1}^{K}\mathbb{E}_{\mathbf{z}\sim P^{k}_{i}}\left[\log\psi(\mathbf{z})_{k}\right].
\label{eq:minimax}
\end{equation}
We realize Eq.~\eqref{eq:minimax} with a gradient reversal layer (GRL) \cite{pmlr-v37-ganin15} placed between $P_i$ and $\psi$, which is the identity in the forward pass and has Jacobian $-\lambda\mathbf{I}$ in the backward pass. A single loss then trains $\psi$ by descent and $P_i$ by scaled ascent:
\begin{equation}
\mathcal{L}_{\mathrm{adv}}
=-\frac{1}{K}\sum\nolimits_{k=1}^{K}
\mathbb{E}_{\mathbf{z}\sim P^{k}_{i}}
\left[\log\psi\!\left(\mathrm{GRL}_{\lambda}(\mathbf{z})\right)_{k}\right].
\label{eq:adv}
\end{equation}
At convergence, embeddings become difficult to distinguish by domain, encouraging the shared projection to preserve recommendation-relevant semantics while discarding domain-specific information. Full theoretical derivation is provided in Appendix.
\subsection{User-Space Unification}
\textbf{User Representation.} Each user is represented by complementary collaborative and semantic signals. The collaborative representation is learned using a LightGCN encoder over the bipartite interaction graph, $e_{u,c}^k=\mathrm{LightGCN}(G^k,u), $
while the semantic representation is obtained by mean-pooling the SBERT
embeddings of the user's interacted items, $e_{u,s}^k=\sum_{i\in\mathcal H_u^k} e_i^k/|\mathcal H_u^k|.$ Then, the two representations are concatenated,
$e_u^k=[e_{u,c}^k;e_{u,s}^k]. $   The collaborative component captures higher-order interaction structure,
whereas the semantic component provides transferable information that
remains available in unseen domains.\\
\textbf{Shared User Projection.} User representations are mapped into the common retrieval space through a shared projection network, $z_u^k = P_u(e_u^k) / \|P_u(e_u^k)\|_2, $  where \(P_u\) is shared across all source domains. Domain-specific collaborative encoders remain outside the shared latent space and are discarded after training, allowing only the transferable projection to generalize to unseen domains. \\
\textbf{Geometric Alignment.} Unlike item representations, aligning user populations cannot rely on domain-adversarial learning as users across domains share neither identities nor one-to-one correspondence~\cite{tang2023robust}.
Instead, transferable user structure is reflected in the geometry of each
population rather than its absolute location in the latent space. We therefore align user populations by minimizing the Gromov--Wasserstein lower bound (GWLB)\cite{Mmoli2011}, which compares the internal pairwise distance structure of two empirical distributions without requiring explicit correspondences.

For a mini-batch of $n$ users from domain $k$ we form the cosine-distance matrix $D^{k}_{ab}=1-\langle\mathbf{z}^{k}_{u,a},\mathbf{z}^{k}_{u,b}\rangle$ and collect its $N=\binom{n}{2}$ strictly upper-triangular entries into the empirical distance measure $\mu^{k}=\tfrac{1}{N}\sum_{a<b}\delta_{D^{k}_{ab}}$. Since $\mu^{a}$ and $\mu^{b}$ are one-dimensional empirical measures of equal size, their squared $2$-Wasserstein distance reduces to a sum over order statistics\cite{peyre2019computational}:
{\small
\begin{equation*}
\mathrm{GWLB}^{2}(a,b)=W_2^{2}(\mu^{a},\mu^{b})
=\frac{1}{N}\sum\nolimits_{r=1}^{N}\!\left(d^{a}_{(r)}-d^{b}_{(r)}\right)^{2},
\label{eq:gwlb}
\end{equation*}}
where $d^{k}_{(r)}$ is the $r$-th smallest pairwise distance in domain $k$. Averaging over domain pairs gives the user-unification loss:
{\small 
\begin{equation}
\mathcal{L}_{\mathrm{gw}}
=\frac{1}{\binom{K}{2}}\sum\nolimits_{1\le a<b\le K}\mathrm{GWLB}^{2}(a,b).
\label{eq:lgw}
\end{equation}}
Eq.~\eqref{eq:lgw} is differentiable and inexpensive to evaluate, making it
practical for mini-batch optimization.
Unlike adversarial alignment, GWLB constrains the intrinsic geometry of
each user population rather than individual embeddings, encouraging
domain-invariant interaction structure while preserving recommendation
semantics. Due to space constraints, detailed background on the adversarial and Gromov--Wasserstein (GW) losses; including full derivations from \cite{ben2010theory} and \cite{Mmoli2011}, is provided in the Appendix. \\

\begin{algorithm}[t]
\caption{Training \textsc{ATLAS}}
\label{alg:training}
\begin{algorithmic}[1]
\REQUIRE Source domains $\{\mathcal{D}^k\}_{k=1}^{K}$; frozen text encoder $E_{\mathrm{text}}$
\ENSURE Projections $P_u, P_i$; codebooks $\{C_\ell\}_{\ell=1}^{L}$

\STATE Compute $e_i^k \leftarrow E_{\mathrm{text}}(t_i)$ for all $i \in \mathcal{I}^k$, $\forall k$
\STATE Pretrain LightGCN on each $\mathcal{G}^k$ to obtain $e_{u,c}^k$ for all $u \in \mathcal{U}^k$

\FOR{each training iteration}
    \STATE Sample $(u, i^+, i^-)$ from each domain $k$; \\ form $e_u^k \leftarrow [e_{u,c}^k\,;\,e_{u,s}^k]$
    \STATE Project and normalize to obtain $z_u^k, z_i^k$ 
    \STATE Compute $\mathcal{L}_{\mathrm{bpr}}$ on $z_u^k, z_i^k$ \eqref{eq:rec};\; $\mathcal{L}_{\mathrm{adv}}$ on $z_i^k$ via GRL \eqref{eq:adv};\; $\mathcal{L}_{\mathrm{gw}}$ on $z_u^k$\ \eqref{eq:lgw}
    \STATE Quantize $z_u^k, z_i^k$ to obtain $\hat{z}_u, \hat{z}_i$ and compute $\mathcal{L}_{\mathrm{rvq}}$   \hfill \eqref{eq:recon} \eqref{eq:rvq}
    \STATE Update $\theta$ by minimizing $\mathcal{L}$ (incl.\ $\mathcal{L}_{\mathrm{aux}}$) \hfill \eqref{eq:total-loss}
\ENDFOR
\STATE freeze $P_u, P_i, \{C_\ell\}$
\end{algorithmic}
\end{algorithm} 
\subsection{Discretizing the Unified Space via RVQ Codebooks}
\label{sec:rvq} 
We discretize the unified space with residual vector quantization (RVQ)~\cite{rajput2023recommender}, using separate quantizers for users and items. RVQ assigns each embedding a hierarchical code and acts as an information bottleneck that removes residual domain-specific variation. Each quantizer has $L$ levels $C_1,\ldots,C_L$, where level $\ell$ holds $M_\ell$ learnable codewords of dimension $d$. Given $\mathbf{z}\in\mathbb{S}^{d-1}$, codes are assigned by quantizing successive residuals ($\mathbf{r}_0=\mathbf{z}$):
\begin{equation*}
c_\ell=\arg\min_{m\le M_\ell}
\lVert\mathbf{r}_{\ell-1}-C_\ell[m]\rVert_2^2,
\quad
\mathbf{r}_\ell=\mathbf{r}_{\ell-1}-C_\ell[c_\ell],
\end{equation*}
yielding the code $\mathbf{c}(\mathbf{z})=(c_1,\ldots,c_L)$ and the $\ell_2$-normalized reconstruction
{\small
\begin{equation}
\hat{\mathbf{z}}
= \Bigl(\sum_\ell C_\ell[c_\ell]\Bigr)
 \Big/
 \Bigl\lVert\sum_\ell C_\ell[c_\ell]\Bigr\rVert_2.
\label{eq:recon}
\end{equation}}
Gradients flow through the quantizer via soft Sinkhorn assignments~\cite{NIPS2013_af21d0c9} during training; hard nearest-codeword assignment is used at inference. \\
\textbf{Quantization losses.}
Following VQ-VAE~\cite{van2017neural}, we use commitment and codebook losses (with $\mathrm{sg}(\cdot)$ denoting stop-gradient) to keep the encoder and codewords mutually consistent:
\begin{equation*}
\mathcal{L}_{\mathrm{commit}}
=\lVert\mathbf{z}-\mathrm{sg}(\hat{\mathbf{z}})\rVert_2^2,
\qquad
\mathcal{L}_{\mathrm{book}}
=\lVert\mathrm{sg}(\mathbf{z})-\hat{\mathbf{z}}\rVert_2^2.
\end{equation*}
A per-level entropy penalty on batch assignment frequencies
$p^{(\ell)}_m$ prevents codebook collapse:
$\mathcal{L}_{\mathrm{ent}}
=
\sum\nolimits_{\ell}\frac{1}{\log M_\ell}
\sum\nolimits_{m} p^{(\ell)}_m \log p^{(\ell)}_m$.
\\
The combined quantization loss is
\begin{equation}
\mathcal{L}_{\mathrm{rvq}}
=\lambda_c\mathcal{L}_{\mathrm{commit}}
+\lambda_b\mathcal{L}_{\mathrm{book}}
+\lambda_e\mathcal{L}_{\mathrm{ent}}.
\label{eq:rvq}
\end{equation}
\textbf{Recommendation loss.} We anchor the space to ranking quality with a sampled-softmax loss. For each positive pair $(u,i^+)$ from domain $k$, we sample $J$ negatives $\{i^-_j\}_{j=1}^J$ from the same domain (masking the user's observed items):
{\small
\begin{equation}
\mathcal{L}_{\mathrm{bpr}}
=-\,\mathbb{E}_{(u,i^{+})\sim\mathcal{R}^{k}}
\log\frac{\exp(\langle\mathbf{z}_u,\mathbf{z}_{i^{+}}\rangle)}{\exp(\langle\mathbf{z}_u,\mathbf{z}_{i^{+}}\rangle)+\sum_{j}\exp(\langle\mathbf{z}_u,\mathbf{z}_{i^{-}_{j}}\rangle)}.
\label{eq:rec}
\end{equation}}
Since all negative samples are drawn exclusively from same domain, cross-domain alignment is driven entirely by $\mathcal{L}_{\mathrm{adv}}$ and $\mathcal{L}_{\mathrm{gw}}$. \\
\textbf{Auxiliary stabilizers.} Adversarial training on a compact manifold is prone to oscillation, so we add three lightweight regularizers: $\mathcal{L}_{\mathrm{centroid}}$ stabilizes per-domain user centroids across steps, $\mathcal{L}_{\mathrm{div}}$ penalizes excessive within-domain user similarity, and $\mathcal{L}_{\mathrm{var}}$ matches per-dimension variance across domains. Their weighted sum is $\mathcal{L}_{\mathrm{aux}}$. These are stabilizers rather than contributions.\\
\begin{table}[t]
\centering
\scriptsize
\setlength{\tabcolsep}{1pt}
\caption{In-domain performance on the five source domains.}
\label{tab:in_domain}
\begin{tabular}{l | cc | cc | cc | cc | cc}
\toprule
& \multicolumn{2}{c|}{\textbf{Beauty}} & \multicolumn{2}{c|}{\textbf{Automotive}} & \multicolumn{2}{c|}{\textbf{Movies \& TV}} & \multicolumn{2}{c|}{\textbf{Video Games}} & \multicolumn{2}{c}{\textbf{Electronics}} \\
\cmidrule(lr){2-3} \cmidrule(lr){4-5} \cmidrule(lr){6-7} \cmidrule(lr){8-9} \cmidrule(lr){10-11}
\textbf{Model} & HR & NDCG & HR & NDCG & HR & NDCG & HR & NDCG & HR & NDCG \\
\midrule
BERT4Rec & .0021 & .0012 & .0013 & .0009 & .0073 & .0042 & .0068 & .0033 & .0024 & .0011 \\
SASRec & .0044 & .0019 & .0014 & .0006 & .0029 & .0013 & .0098 & .0054 & .0022 & .0008 \\
LightGCN & .0008 & .0005 & .0006 & .0004 & .0131 & .0072 & .0182 & .0091 & .0017 & .0009 \\
PrepRec & .0140 & .0071 & .0233 & .0138 & .0221 & .0116 & .0342 & .0201 & .0221 & .0128 \\
LLM-RecG & .0257 & .0150 & .0222 & .0113 & .0357 & .0240 & .0620 & .0412 & .0199 & .0102 \\
UniSRec & \underline{.0141} & \underline{.0073} & \underline{.0127} & \underline{.0068} & \underline{.0770} & \underline{.0041} & \underline{.1439} & \underline{.1035} & \underline{.0213} & \underline{.0112} \\
\midrule
ATLAS & \textbf{.0870} & \textbf{.0490} & \textbf{.0532} & \textbf{.0332} & \textbf{.1388} & \textbf{.0774} & \textbf{.2116} & \textbf{.1561} & \textbf{.0687} & \textbf{.0443} \\
\bottomrule
\end{tabular}
\end{table}
\textbf{Overall objective.} Training minimizes
\begin{multline}
\label{eq:total-loss}
    \mathcal{L}
    =\lambda_{\mathrm{bpr}}\mathcal{L}_{\mathrm{bpr}}
    +\lambda_{\mathrm{adv}}\mathcal{L}_{\mathrm{adv}}\\ 
    +\lambda_{\mathrm{gw}}\mathcal{L}_{\mathrm{gw}}
    +\lambda_{\mathrm{rvq}}\mathcal{L}_{\mathrm{rvq}}
    +\lambda_{\mathrm{aux}}\mathcal{L}_{\mathrm{aux}},
\end{multline}
whose terms are complementary: $\mathcal{L}_{\mathrm{bpr}}$ makes the space predictive, $\mathcal{L}_{\mathrm{adv}}$ makes item marginals domain-indistinguishable, $\mathcal{L}_{\mathrm{gw}}$ encourages source domains to exhibit similar distributions of within-domain user distances, and $\mathcal{L}_{\mathrm{rvq}}$ makes the space discretizable. Optimizing Eq.~\eqref{eq:total-loss} completes model training; the LightGCN encoders and the discriminator $\psi$ are then discarded, and the projections $P_u,\, P_i$ and codebooks $\{C_{\ell}\}$ are frozen. Full hyperparameters, training infrastructure, and implementation details are provided in Appendix.
\subsection{Zero-Shot Inference on an Unseen Domain}
\label{sec:inference}
Given the target domain $\mathcal{D}_{\mathrm{new}}$, inference proceeds in three steps with no parameter update. \\
\textbf{Projection with frozen encoders.}
LightGCN embeddings cannot be obtained for $\mathcal{D}_{\mathrm{new}}$:
they are graph-specific, requiring newly learned embedding tables for
disjoint target users and items, which would violate RDG's
frozen-model constraint. Users are therefore initialized through the
semantic channel alone; this mean-pooled vector serves only as an
initialization, with the frozen projection $P_u$ and codebooks mapping
it into the shared retrieval space (modality-dropout ablation in
Appendix). Projection vectors are computed as follows: $\mathbf{z}^{\mathrm{new}}_{i}=\frac{P_i(E_{\mathrm{text}}(t_i))}{\lVert P_i(E_{\mathrm{text}}(t_i))\rVert_2}$ and 
$\mathbf{z}^{\mathrm{new}}_{u}=\frac{P_u\!\left(\left[\mathbf{0}\,;\,\mathbf{e}^{\mathrm{new}}_{u,\mathrm{sem}}\right]\right)}{\lVert P_u\!\left(\left[\mathbf{0}\,;\,\mathbf{e}^{\mathrm{new}}_{u,\mathrm{sem}}\right]\right)\rVert_2},$
where $\mathbf{e}^{\mathrm{new}}_{u,\mathrm{sem}}$ mean-pools the SBERT embeddings of the target user's history $\mathcal{H}^{\mathrm{new}}_{u}$. \\
\textbf{Quantization and reconstruction.} Both embeddings are
passed through the frozen codebooks with hard nearest-code
assignment, yielding semantic identifiers $c(z^{\mathrm{new}}_{u}), c(z^{\mathrm{new}}_{i})$
and the reconstructed vectors
{\small
\begin{equation*}
\hat{z}^{\mathrm{new}}_{u}
=\Pi\!\left(\sum_{\ell=1}^{L}C^{u}_{\ell}\!\left[c_{\ell}(z^{\mathrm{new}}_{u})\right]\right), \, \,
\hat{z}^{\mathrm{new}}_{i}
= \Pi\!\left(\sum_{\ell=1}^{L}C^{i}_{\ell}\!\left[c_{\ell}(z^{\mathrm{new}}_{i})\right]\right),
\end{equation*}}
where $\Pi(\cdot)$ denotes $\ell_2$-normalization, matching the
reconstruction in Eq.~\eqref{eq:recon}. Because the codebooks were
fitted on the unified space, they tokenize the unseen domain without
adaptation; reconstruction additionally snaps target embeddings onto
the source-supported region of the sphere, discarding any residual
target-specific direction the projections may have retained.\\
\textbf{Top-$K$ ranking.} Candidates are scored by the inner product of the reconstructed vectors, $s(u,i)=\big\langle\hat{\mathbf{z}}^{\mathrm{new}}_{u},\,\hat{\mathbf{z}}^{\mathrm{new}}_{i}\big\rangle,$
and the $K$ highest-scoring items are returned. As both vectors are unit-norm, retrieval reduces to maximum inner-product search over a codebook-indexed item set and is served by standard ANN structures. Since no target interaction is used for any parameter update or model selection, and no user or item identity is shared with the source domains, this constitutes strict zero-shot transfer.

\begin{table*}[t]
\centering
\scriptsize
\setlength{\tabcolsep}{1.8pt}
\renewcommand{\arraystretch}{1.15}
\caption{Zero-shot and fine-tuned recommendation performance on ten unseen target domains. Absolute values are low due to
full-ranking retrieval over the complete item catalogue (up to
hundreds of thousands of items); this holds uniformly across all
methods, so relative comparisons remain valid (dataset statistics in
Appendix).}
\label{tab:zero_shot}
\begin{tabular}{l | cc | cc | cc | cc | cc | cc | cc | cc | cc | cc}
\toprule
 & \multicolumn{2}{c|}{\textbf{Dig. Music}} 
 & \multicolumn{2}{c|}{\textbf{Cell Phones}} 
 & \multicolumn{2}{c|}{\textbf{Arts \& Crafts}} 
 & \multicolumn{2}{c|}{\textbf{Baby Prod.}} 
 & \multicolumn{2}{c|}{\textbf{Books}} 
 & \multicolumn{2}{c|}{\textbf{Music Instr.}} 
 & \multicolumn{2}{c|}{\textbf{Office Prod.}} 
 & \multicolumn{2}{c|}{\textbf{Appliances}} 
 & \multicolumn{2}{c|}{\textbf{Software}} 
 & \multicolumn{2}{c}{\textbf{Pet Supplies}} \\
\cmidrule(lr){2-3} \cmidrule(lr){4-5} \cmidrule(lr){6-7} \cmidrule(lr){8-9} \cmidrule(lr){10-11} \cmidrule(lr){12-13} \cmidrule(lr){14-15} \cmidrule(lr){16-17} \cmidrule(lr){18-19} \cmidrule(lr){20-21}
\textbf{Model} & HR & NDCG & HR & NDCG & HR & NDCG & HR & NDCG & HR & NDCG & HR & NDCG & HR & NDCG & HR & NDCG & HR & NDCG & HR & NDCG \\

\midrule
ZESRec 
 & .0102 & .0059 & .0002 & .0001 & .0022 & .0014 & .0017 & .0009 & .0013 & .0006 & .0016 & .0008 & .0010 & .0005 & .0024 & .0010 & .0014 & .0007 & .0005 & .0002 \\
VQ-Rec 
 & .0113 & .0053 & \underline{.0013} & \underline{.0009} & \underline{.0043} & \underline{.0031} & \underline{.0029} & \underline{.0018} & .0017 & .0009 & \underline{.0039} & \underline{.0027} & \underline{.0013} & \underline{.0013} & \underline{.0092} & \underline{.0047} & .0038 & .0018 & \textbf{.0019} & \textbf{.0015} \\
LLM-RecG 
 & .0131 & .0077 & .0011 & .0007 & .0008 & .0005 & .0019 & .0011 & .0023 & .0014 & .0031 & .0020 & .0017 & .0009 & .0067 & .0042 & \underline{.0053} & \underline{.0029} & .0009 & .0005 \\
UniSRec$_\text{ZS}$ 
 & \textbf{.0206} & \textbf{.0103} & .0009 & .0005 & .0039 & .0021 & .0028 & .0014 & \underline{.0029} & .0012 & .0028 & .0015 & \textbf{.0030} & \textbf{.0016} & .0078 & .0043 & .0049 & .0026 & \underline{.0011} & \underline{.0005} \\
ATLAS$_\text{ZS}$ (ours)
 & \underline{.0173} & \underline{.0089} & \textbf{.0023} & \textbf{.0012} & \textbf{.0067} & \textbf{.0035} & \textbf{.0042} & \textbf{.0022} & \textbf{.0055} & \textbf{.0027} & \textbf{.0061} & \textbf{.0038} & .0026 & .0012 & \textbf{.0113} & \textbf{.0059} & \textbf{.0073} & \textbf{.0043} & .0009 & .0003 \\

\midrule
GWCDR
 & .0164 & .0091 & .0009 & .0005 & .0011 & .0009 & .0017 & .0010 & .0011 & .0007 & .0010 & .0007 & .0006 & .0004 & .0045 & .0029 & .0063 & .0038 & .0007 & .0004 \\
PrepRec 
 & .0191 & .0092 & .0031 & .0021 & .0028 & .0017 & .0050 & .0025 & \underline{.0058} & \underline{.0032} & .0076 & .0040 & .0040 & .0024 & .0085 & .0042 & .0096 & .0051 & .0096 & .0059 \\
UniSRec$_\text{FT}$ 
 & \underline{.0291} & \underline{.0140} & .0030 & .0017 & \underline{.0051} & \underline{.0031} & \underline{.0055} & \underline{.0027} & .0055 & .0029 & \underline{.0192} & \underline{.0097} & \underline{.0121} & \textbf{.0073} & \underline{.0168} & \underline{.0081} & \underline{.0484} & \underline{.0251} & \underline{.0112} & \underline{.0058} \\
ATLAS$_\text{FT}$ (ours)
 & \textbf{.0297} & \textbf{.0152} & \textbf{.0075} & \textbf{.0037} & \textbf{.0094} & \textbf{.0051} & \textbf{.0077} & \textbf{.0038} & \textbf{.0103} & \textbf{.0049} & \textbf{.0331} & \textbf{.0176} & \textbf{.0144} & \underline{.0070} & \textbf{.0442} & \textbf{.0242} & \textbf{.0601} & \textbf{.0411} & \textbf{.0197} & \textbf{.0099} \\
\bottomrule
\end{tabular}
\end{table*}

\section{Experiments}
\label{sec:experiments}

\begin{table*}[t]
\centering
\scriptsize
\setlength{\tabcolsep}{2.5pt}

\begin{minipage}[t]{0.48\textwidth}
\centering
\caption{Ablation on alignment objectives on target domains.}
\label{tab:ablation_loss}
\vspace{2pt}
\begin{tabular}{l | cccccccccc}
\toprule
\textbf{Configuration}
& \textbf{Cell}
& \textbf{Musi.}
& \textbf{Arts}
& \textbf{Baby}
& \textbf{Book}
& \textbf{Inst.}
& \textbf{Offi.}
& \textbf{Appl.}
& \textbf{Soft.}
& \textbf{Pet} \\
\midrule
BPR only
& .0004 & .0061 & .0018 & .0014 & .0020
& .0029 & .0011 & .0041 & .0021 & .0004 \\

BPR + $\mathcal{L}_\mathrm{adv}$
& .0002 & .0072 & .0053 & .0009 & .0042
& .0044 & .0020 & .0066 & .0061 & .0007 \\

BPR + $\mathcal{L}_\mathrm{GW}$
& .0009 & .0107 & .0021 & .0032 & .0033
& .0041 & .0015 & .0091 & .0036 & .0004 \\
\midrule

ATLAS (full)
& \textbf{.0023}
& \textbf{.0173}
& \textbf{.0067}
& \textbf{.0042}
& \textbf{.0055}
& \textbf{.0061}
& \textbf{.0026}
& \textbf{.0113}
& \textbf{.0073}
& \textbf{.0009} \\
\bottomrule
\end{tabular}
\end{minipage}%
\hfill%
\begin{minipage}[t]{0.48\textwidth}
\centering
\caption{Source-domain diversity effect on target domains.}
\label{tab:source_diversity}
\vspace{2pt}
\begin{tabular}{c | cccccccccc}
\toprule
\textbf{$K$}
& \textbf{Cell}
& \textbf{Musi.}
& \textbf{Arts}
& \textbf{Baby}
& \textbf{Book}
& \textbf{Inst.}
& \textbf{Offi.}
& \textbf{Appl.}
& \textbf{Soft.}
& \textbf{Pet} \\
\midrule
2
& .0005 & .0071 & .0011 & .0005 & .0004
& .0008 & .0001 & .0027 & .0013 & .0004 \\

3
& .0009 & .0112 & .0020 & .0014 & .0016
& .0017 & .0007 & .0057 & .0028 & .0002 \\

4
& .0015 & .0136 & .0040 & .0034 & .0038
& .0041 & .0014 & .0072 & .0056 & .0007 \\

5
& \textbf{.0023}
& \textbf{.0173}
& \textbf{.0067}
& \textbf{.0042}
& \textbf{.0055}
& \textbf{.0061}
& \textbf{.0026}
& \textbf{.0113}
& \textbf{.0073}
& \textbf{.0009} \\
\bottomrule
\end{tabular}
\end{minipage}

\end{table*}

We design our experiments to answer four research questions:
\textbf{RQ1}: Does ATLAS generalize to entirely unseen domains without target-domain adaptation, and how does it compare against sequential, graph-based, cross-domain, pre-training, universal, quantization-based, and LLM-based baselines? 
\textbf{RQ2}: Is each alignment objective necessary? 
\textbf{RQ3}: Does codebook-reconstructed $\mathbf{z}_\mathrm{rec}$ outperform the continuous projection $\mathbf{z}_\mathrm{proj}$? 
\textbf{RQ4}: How does source-domain diversity affect zero-shot transfer?
\subsection{Experimental Setup}
\label{sec:setup}
\textbf{Datasets.} We use 15 domains from the Amazon Reviews 2023 dataset,~\cite{hou2024bridging} partitioned into five \emph{source} domains for training and ten held-out \emph{target} domains for zero-shot evaluation. Users with fewer than five interactions are removed. Evaluation follows full-ranking leave-one-out.
Detailed dataset statistics are provided in Appendix. \\
\textbf{Training and Evaluation Protocol.}
We adopt the leave-one-out split: the last interaction per user is the test item, and the rest are training. All evaluations use \emph{full ranking} against the entire candidate pool without negative sampling. At zero-shot evaluation, the frozen ATLAS model is applied to each target domain without parameter updates. 
Full training configuration is provided in Appendix. \\
\textbf{Metrics.} For retrieval-based baselines, we report HR@10 and NDCG@10 under full ranking. For LLM-based binary scoring baselines producing pointwise decisions, we report AUC and HR@1 following their original protocols.\\
\textbf{Baselines.} \emph{Same-domain}: BERT4Rec, SASRec, LightGCN, PrepRec, LLM-RecG~\cite{li2025llm}. 
\emph{Unseen domains}: ZESRec, VQ Rec~\cite{hou2023vqrec}, PrepRec, GWCDR\cite{li2022gromov} , UniSRec, and ATLAS. 
\emph{LLM binary scoring}: LLMRec, TALLRec and BigRec \cite{10.1145/3716393}.

\begin{table*}[t]
\centering
\caption{Performance comparison of ATLAS against LLM baselines.}
\label{tab:model_results}
\small
\setlength{\tabcolsep}{1.2pt}
\begin{tabular}{l | cc | cc | cc | cc | cc | cc | cc | cc | cc | cc}
\toprule
& \multicolumn{2}{c|}{\textbf{Musi.}} & \multicolumn{2}{c|}{\textbf{Cell}} & \multicolumn{2}{c|}{\textbf{Book}} & \multicolumn{2}{c|}{\textbf{Arts}} & \multicolumn{2}{c|}{\textbf{Baby}} & \multicolumn{2}{c|}{\textbf{Pets}} & \multicolumn{2}{c|}{\textbf{Appl.}} & \multicolumn{2}{c|}{\textbf{Offi.}} & \multicolumn{2}{c|}{\textbf{Soft.}} & \multicolumn{2}{c}{\textbf{Inst.}} \\
\cmidrule(lr){2-3} \cmidrule(lr){4-5} \cmidrule(lr){6-7} \cmidrule(lr){8-9} \cmidrule(lr){10-11} \cmidrule(lr){12-13} \cmidrule(lr){14-15} \cmidrule(lr){16-17} \cmidrule(lr){18-19} \cmidrule(lr){20-21}
\textbf{Model} & \textbf{AUC} & \textbf{HR1} & \textbf{AUC} & \textbf{HR1} & \textbf{AUC} & \textbf{HR1} & \textbf{AUC} & \textbf{HR1} & \textbf{AUC} & \textbf{HR1} & \textbf{AUC} & \textbf{HR1} & \textbf{AUC} & \textbf{HR1} & \textbf{AUC} & \textbf{HR1} & \textbf{AUC} & \textbf{HR1} & \textbf{AUC} & \textbf{HR1} \\
\midrule
LLMRec. & .55 & .13 & .51 & .09 & .51 & .11 & .55 & .10 & .55 & .13 & .53 & .12 & .58 & .13 & .56 & .10 & .57 & .12 & .55 & .12 \\
BigRec    & .67 & .19 & .61 & .15 & .64 & .13 & .64 & .15 & .62 & .14 & .61 & .13 & .61 & .19 & .62 & .14 & .64 & .19 & \textbf{.66} & \textbf{.17} \\
TallRec   & .60 & .15 & .58 & .13 & .64 & .12 & .61 & .14 & \textbf{.62} & \textbf{.15} & \textbf{.69} & \textbf{.17} & .64 & .17 & \textbf{.67} & \textbf{.15} & .64 & .17 & .63 & .16 \\
ATLAS  & \textbf{.71} & \textbf{.26} & \textbf{.68} & \textbf{.18} & \textbf{.67} & \textbf{.17} & \textbf{.71} & \textbf{.20} & .57 & .11 & .66 & .15 & \textbf{.67} & \textbf{.25} & .63 & .13 & \textbf{.69} & \textbf{.22} & \textbf{.66} & .16 \\
\bottomrule
\end{tabular}
\end{table*}

\begin{figure}[t]
  \centering
  \begin{minipage}[t]{0.49\linewidth}
    \centering
    \includegraphics[width=\linewidth]{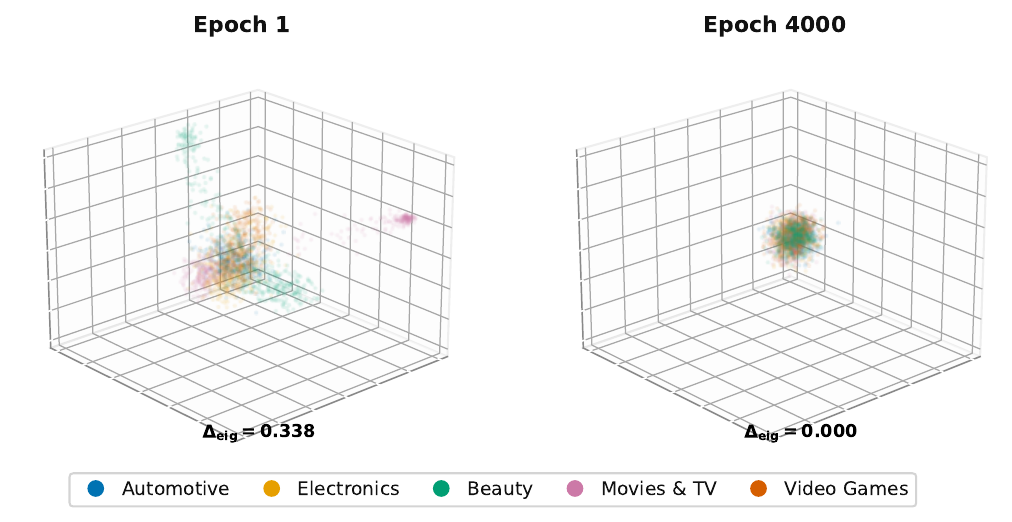}
    \caption{User embedding alignment}
    \label{fig:user_alignment}
  \end{minipage}
  \hfill
    \begin{minipage}[t]{0.49\linewidth}
    \centering
    \includegraphics[width=\linewidth]{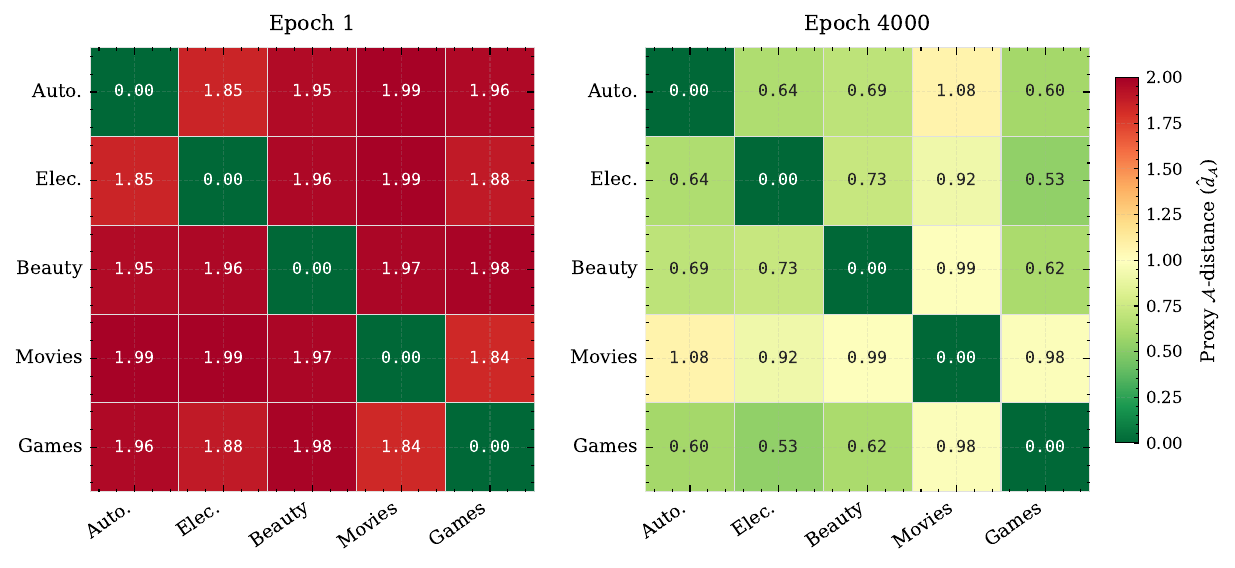}
    \caption{PAD across source domains}
    \label{fig:pad_matrix}
  \end{minipage}
\end{figure}
\subsection{RQ1: Generalization to unseen domains}
\label{sec:rq1}
\textbf{In-Domain Performance (Source Domains).}
Table~\ref{tab:in_domain} reports performance on the five source domains. Since SASRec, BERT4Rec, and related baselines were designed for domain-specific recommendation, we adapt them to the RDG setting by jointly training on a unified interaction pool obtained by merging the five source domains~\cite{10.1145/3715099}.
First, ATLAS consistently achieves the strongest overall performance, demonstrating that enforcing domain-invariant representations does not compromise recommendation accuracy on the training domains.
Second, the strongest competing baseline is generally UniSRec, indicating that large-scale recommendation pre-training provides a strong foundation for transfer.
Finally, the improvement over domain-specific methods such as SASRec and LightGCN confirms that jointly exploiting multiple heterogeneous recommendation domains produces representations that remain competitive even within the seen domains. \\
\textbf{Zero-Shot Transfer (Unseen Target Domains).} 
Table~\ref{tab:zero_shot} presents zero-shot evaluation on ten
completely unseen domains. ATLAS$_\text{ZS}$ outperforms all
zero-shot baselines on seven of ten target domains, and remains
competitive on the remaining three.\footnote{Wilcoxon tests: ATLAS$_\text{ZS}$ beats all
zero-shot baselines ($p<0.05$) except UniSRec$_\text{ZS}$ ($p{=}0.13$,
$p{=}0.05$); ATLAS$_\text{FT}$ beats all fine-tuned baselines
($p<0.01$). Seed runs in Appendix.}
This demonstrates that the learned retrieval space generalizes to unseen domains without target-domain adaptation, though the advantage is not uniform across all domains.
For reference, we additionally report fine-tuned variants: UniSRec$_\text{FT}$
and ATLAS$_\text{FT}$ are further fine-tuned on each target domain.
GWCDR, PrepRec, and LLM-RecG are pairwise cross-domain methods, trained
separately for each (source, target) domain pair; for each target domain
we report the average over its five source-domain pairings.
Here, ATLAS$_\text{FT}$ achieves the strongest overall performance among them.
The zero-shot performance of ATLAS therefore supports the central hypothesis of this work: recommendation knowledge learned from multiple heterogeneous domains can transfer effectively to most of the evaluated unseen domains, without requiring target-domain adaptation.  \\
\textbf{Comparison with LLM-based Binary Scoring.}
Table~\ref{tab:model_results} compares ATLAS against LLM-based
recommenders under their native binary-scoring protocol: since LLM
context length cannot accommodate full-catalogue ranking, all methods; including ATLAS, are evaluated over a reduced pool of 20
items (1 positive + 19 negatives). ATLAS achieves higher AUC
and HR@1 on most domains despite relying solely on
recommendation-specific representation learning, without
language-model inference at deployment.
\subsection{RQ2: Ablation on Alignment Objectives}
\label{sec:rq2}
Table~\ref{tab:ablation_loss} shows target-domain results for each
alignment configuration (full results including source-domain
performance are provided in Appendix). Either
alignment objective improves zero-shot recommendation over BPR alone,
while their combination performs best across the target domains. The
gains are complementary: adversarial alignment reduces item-domain
separability, whereas Gromov--Wasserstein alignment encourages
comparable user-population geometry, as examined further in the
representational analysis. The full model therefore benefits from
addressing both item-side and user-side domain shift.
\subsection{RQ3: Codebook vs.\ Projected Embeddings}
\label{sec:rq3}
To isolate the discrete bottleneck's contribution, we train two
variants: ATLAS without RVQ (scored via $z_{\mathrm{proj}}$)
and full ATLAS with RVQ (scored via $z_{\mathrm{rec}}$).
$z_{\mathrm{rec}}$ consistently outperforms $z_{\mathrm{proj}}$ across
all target domains, improving HR@10 by 45\% on average, showing that
RVQ codebooks improve zero-shot transfer beyond the continuous
alignment losses alone, by suppressing source-specific variation while
preserving transferable structure. Detailed results are in
Appendix.
\subsection{RQ4: Source-Domain Diversity Effect}
\label{sec:rq4}
Table~\ref{tab:source_diversity} shows target-domain zero-shot performance as
the number of source domains increases (source-domain performance at
each $K$ is provided in Appendix). The
monotonic trend on target domains suggests that exposure to more
heterogeneous recommendation environments improves the transferability
of the learned representations, supporting the motivation behind RDG.
\subsection{Representational Analysis}

We examine whether the learned embeddings exhibit the intended invariances. Additional analyses and diagnostics are provided in the Appendix and Media Supplement.\\
\textbf{User-space alignment.}
Figure~\ref{fig:user_alignment} projects 500 users per source domain
onto the top-3 PCA components. Domain
clusters, initially separated, progressively overlap during training; consistent with $\mathcal{L}_{\mathrm{GW}}$ aligning relational
geometry without user correspondences.\\
\textbf{Item-domain invariance.}
Figure~\ref{fig:pad_matrix} reports the pairwise Proxy
$\mathcal{A}$-distance (PAD) between source-domain item representations,
$\hat d_{\mathcal A}(k,k')
=
2\bigl(1-2\epsilon_{\mathcal H}(k,k')\bigr),$
where $\epsilon_{\mathcal H}$ is the error of a linear domain classifier. Lower PAD indicates reduced linear separability
between domains. PAD decreases substantially during training for most
domain pairs, showing that a linear classifier becomes less able to
recover the source domain. This provides direct empirical evidence that
adversarial training reduces linearly domain-identifiable information in
the projected item representations.

\section{Conclusion}
\label{sec:conclusion}
We introduce RDG and propose ATLAS, a multi-source retrieval framework that learns transferable user--item representations through complementary semantic and geometric alignment with shared hierarchical RVQ. Experiments on five source and ten unseen Amazon domains show that ATLAS outperforms strong domain-specific, cross-domain, universal, and LLM-based baselines on most unseen domains under zero-shot transfer. Ablation studies and representation analyses further support the proposed alignment mechanisms, while source-domain diversity experiments suggest that broader training environments improve transferability. Overall, our results indicate that robust recommendation knowledge can emerge directly from recommendation-specific learning across heterogeneous domains, without target-domain adaptation or language-model pretraining. This work lays the foundation for RDG as a benchmark for developing recommendation models that generalize reliably across previously unseen recommendation environments.

\clearpage

\bibliography{aaai2027}

@InProceedings{pmlr-v37-ganin15,
  title = 	 {Unsupervised Domain Adaptation by Backpropagation},
  author = 	 {Ganin, Yaroslav and Lempitsky, Victor},
  booktitle = 	 {Proceedings of the 32nd International Conference on Machine Learning},
  pages = 	 {1180--1189},
  year = 	 {2015},
  editor = 	 {Bach, Francis and Blei, David},
  volume = 	 {37},
  series = 	 {Proceedings of Machine Learning Research},
  address = 	 {Lille, France},
  month = 	 {07--09 Jul},
  publisher =    {PMLR},
  url = 	 {https://proceedings.mlr.press/v37/ganin15.html},
}

@article{ben2010theory,
  title={A theory of learning from different domains},
  author={Ben-David, Shai and Blitzer, John and Crammer, Koby and Kulesza, Alex and Pereira, Fernando and Vaughan, Jennifer Wortman},
  journal={Machine learning},
  volume={79},
  number={1},
  pages={151--175},
  year={2010},
  publisher={Springer}
}

@misc{goodfellow2014generativeadversarialnetworks,
  title={Generative Adversarial Networks},
  author={Goodfellow, Ian J. and Pouget-Abadie, Jean and Mirza, Mehdi and Xu, Bing and Warde-Farley, David and Ozair, Sherjil and Courville, Aaron and Bengio, Yoshua},
  year={2014},
  eprint={1406.2661},
  archivePrefix={arXiv},
  primaryClass={stat.ML},
  url={https://arxiv.org/abs/1406.2661}
}

@inproceedings{pei2018multi,
  title={Multi-adversarial domain adaptation},
  author={Pei, Zhongyi and Cao, Zhangjie and Long, Mingsheng and Wang, Jianmin},
  booktitle={Proceedings of the AAAI Conference on Artificial Intelligence},
  volume={32},
  year={2018}
}

@article{Mmoli2011,
  title = {Gromov–Wasserstein Distances and the Metric Approach to Object Matching},
  volume = {11},
  ISSN = {1615-3383},
  url = {http://dx.doi.org/10.1007/s10208-011-9093-5},
  DOI = {10.1007/s10208-011-9093-5},
  number = {4},
  journal = {Foundations of Computational Mathematics},
  publisher = {Springer Science and Business Media LLC},
  author = {Mémoli,  Facundo},
  year = {2011},
  month = Apr,
  pages = {417–487}
}

@inproceedings{lyu2024llm,
  title={Llm-rec: Personalized recommendation via prompting large language models},
  author={Lyu, Hanjia and Jiang, Song and Zeng, Hanqing and Xia, Yinglong and Wang, Qifan and Zhang, Si and Chen, Ren and Leung, Chris and Tang, Jiajie and Luo, Jiebo},
  booktitle={Findings of the Association for Computational Linguistics: NAACL 2024},
  pages={583--612},
  year={2024}
}

@article{van2017neural,
  title={Neural discrete representation learning},
  author={Van Den Oord, Aaron and Vinyals, Oriol and others},
  journal={Advances in neural information processing systems},
  volume={30},
  year={2017}
}

@book{peyre2019computational,
  title={Computational optimal transport: With applications to data science},
  author={Peyr{\'e}, Gabriel and Cuturi, Marco},
  year={2019},
  publisher={Now Foundations and Trends}
}

@inproceedings{NIPS2013_af21d0c9,
 author = {Cuturi, Marco},
 booktitle = {Advances in Neural Information Processing Systems},
 editor = {C.J. Burges and L. Bottou and M. Welling and Z. Ghahramani and K. Weinberger},
 pages = {},
 publisher = {Curran Associates, Inc.},
 title = {Sinkhorn Distances: Lightspeed Computation of Optimal Transport},
 url = {https://proceedings.neurips.cc/paper_files/paper/2013/file/af21d0c97db2e27e13572cbf59eb343d-Paper.pdf},
 volume = {26},
 year = {2013}
}

@inproceedings{hou2023vqrec,
  author = {Yupeng Hou and Zhankui He and Julian McAuley and Wayne Xin Zhao},
  title = {Learning Vector-Quantized Item Representation for Transferable Sequential Recommenders},
  booktitle={{TheWebConf}},
  year = {2023}
}

@inproceedings{kang2018self,
  title={Self-attentive sequential recommendation},
  author={Kang, Wang-Cheng and McAuley, Julian},
  booktitle={2018 IEEE international conference on data mining (ICDM)},
  pages={197--206},
  year={2018},
  organization={IEEE}
}

@article{hou2024bridging,
  title={Bridging Language and Items for Retrieval and Recommendation},
  author={Hou, Yupeng and Li, Jiacheng and He, Zhankui and Yan, An and Chen, Xiusi and McAuley, Julian},
  journal={arXiv preprint arXiv:2403.03952},
  year={2024}
}

@inproceedings{tang2023robust,
  title={Robust attributed graph alignment via joint structure learning and optimal transport},
  author={Tang, Jianheng and Zhang, Weiqi and Li, Jiajin and Zhao, Kangfei and Tsung, Fugee and Li, Jia},
  booktitle={2023 IEEE 39th International Conference on Data Engineering (ICDE)},
  pages={1638--1651},
  year={2023},
  organization={IEEE}
}

@inproceedings{he2020lightgcn,
  title={Lightgcn: Simplifying and powering graph convolution network for recommendation},
  author={He, Xiangnan and Deng, Kuan and Wang, Xiang and Li, Yan and Zhang, Yongdong and Wang, Meng},
  booktitle={Proceedings of the 43rd International ACM SIGIR conference on research and development in Information Retrieval},
  pages={639--648},
  year={2020}
}

@inproceedings{sun2019bert4rec,
  title={BERT4Rec: Sequential recommendation with bidirectional encoder representations from transformer},
  author={Sun, Fei and Liu, Jun and Wu, Jian and Pei, Changhua and Lin, Xiao and Ou, Wenwu and Jiang, Peng},
  booktitle={Proceedings of the 28th ACM international conference on information and knowledge management},
  pages={1441--1450},
  year={2019}
}

@inproceedings{hou2022towards,
  title={Towards universal sequence representation learning for recommender systems},
  author={Hou, Yupeng and Mu, Shanlei and Zhao, Wayne Xin and Li, Yaliang and Ding, Bolin and Wen, Ji-Rong},
  booktitle={Proceedings of the 28th ACM SIGKDD conference on knowledge discovery and data mining},
  pages={585--593},
  year={2022}
}

@inproceedings{wang2024pre,
  title={A pre-trained zero-shot sequential recommendation framework via popularity dynamics},
  author={Wang, Junting and Rathi, Praneet and Sundaram, Hari},
  booktitle={Proceedings of the 18th ACM Conference on Recommender Systems},
  pages={433--443},
  year={2024}
}

@inproceedings{bao2023tallrec,
  title={Tallrec: An effective and efficient tuning framework to align large language model with recommendation},
  author={Bao, Keqin and Zhang, Jizhi and Zhang, Yang and Wang, Wenjie and Feng, Fuli and He, Xiangnan},
  booktitle={Proceedings of the 17th ACM conference on recommender systems},
  pages={1007--1014},
  year={2023}
}

@inproceedings{li2025llm,
  title={Llm-recg: A semantic bias-aware framework for zero-shot sequential recommendation},
  author={Li, Yunzhe and Wang, Junting and Sundaram, Hari and Liu, Zhining},
  booktitle={Proceedings of the Nineteenth ACM Conference on Recommender Systems},
  pages={237--246},
  year={2025}
}

@article{10.1145/3715099,
author = {Xin, Haoran and Sun, Ying and Wang, Chao and Xiong, Hui},
title = {LLMCDSR: Enhancing Cross-Domain Sequential Recommendation with Large Language Models},
year = {2025},
issue_date = {September 2025},
publisher = {Association for Computing Machinery},
address = {New York, NY, USA},
volume = {43},
number = {5},
issn = {1046-8188},
url = {https://doi.org/10.1145/3715099},
doi = {10.1145/3715099},
journal = {ACM Trans. Inf. Syst.},
month = jul,
articleno = {120},
numpages = {33}
}

@article{ding2021zero,
  title={Zero-shot recommender systems},
  author={Ding, Hao and Ma, Yifei and Deoras, Anoop and Wang, Yuyang and Wang, Hao},
  journal={arXiv preprint arXiv:2105.08318},
  year={2021}
}

@article{rajput2023recommender,
  title={Recommender systems with generative retrieval},
  author={Rajput, Shashank and Mehta, Nikhil and Singh, Anima and Hulikal Keshavan, Raghunandan and Vu, Trung and Heldt, Lukasz and Hong, Lichan and Tay, Yi and Tran, Vinh and Samost, Jonah and others},
  journal={Advances in Neural Information Processing Systems},
  volume={36},
  pages={10299--10315},
  year={2023}
}

@inproceedings{hu2018conet,
  title={Conet: Collaborative cross networks for cross-domain recommendation},
  author={Hu, Guangneng and Zhang, Yu and Yang, Qiang},
  booktitle={Proceedings of the 27th ACM international conference on information and knowledge management},
  pages={667--676},
  year={2018}
}

@inproceedings{man2017cross,
  title={Cross-domain recommendation: An embedding and mapping approach.},
  author={Man, Tong and Shen, Huawei and Jin, Xiaolong and Cheng, Xueqi},
  booktitle={Ijcai},
  volume={17},
  pages={2464--2470},
  year={2017}
}

@inproceedings{inproceedings,
  author = {Elkahky, Ali Mamdouh and Song, Yang and He, Xiaodong},
  title = {A Multi-View Deep Learning Approach for Cross Domain User Modeling in Recommendation Systems},
  booktitle = {Proceedings of the 24th International Conference on World Wide Web (WWW)},
  pages = {278--288},
  year = {2015},
  doi = {10.1145/2736277.2741667}
}

@inproceedings{li2022gromov,
  title={Gromov-wasserstein guided representation learning for cross-domain recommendation},
  author={Li, Xinhang and Qiu, Zhaopeng and Zhao, Xiangyu and Wang, Zihao and Zhang, Yong and Xing, Chunxiao and Wu, Xian},
  booktitle={Proceedings of the 31st ACM International Conference on Information \& Knowledge Management},
  pages={1199--1208},
  year={2022}
}

@inproceedings{liu2022exploiting,
  title={Exploiting variational domain-invariant user embedding for partially overlapped cross domain recommendation},
  author={Liu, Weiming and Zheng, Xiaolin and Su, Jiajie and Hu, Mengling and Tan, Yanchao and Chen, Chaochao},
  booktitle={Proceedings of the 45th International ACM SIGIR conference on research and development in information retrieval},
  pages={312--321},
  year={2022}
}

@inproceedings{xiao2026modeling,
  title={Modeling User Preferences as Distributions for Optimal Transport-Based Cross-Domain Recommendation under Non-overlapping Settings},
  author={Xiao, Ziyin and Suzumura, Toyotaro},
  booktitle={Pacific-Asia Conference on Knowledge Discovery and Data Mining},
  pages={447--458},
  year={2026},
  organization={Springer}
}

@inproceedings{lin2024pre,
  title={Pre-trained recommender systems: A causal debiasing perspective},
  author={Lin, Ziqian and Ding, Hao and Hoang, Nghia Trong and Kveton, Branislav and Deoras, Anoop and Wang, Hao},
  booktitle={Proceedings of the 17th ACM International Conference on Web Search and Data Mining},
  pages={424--433},
  year={2024}
}

@inproceedings{geng2022recommendation,
  title={Recommendation as language processing (rlp): A unified pretrain, personalized prompt \& predict paradigm (p5)},
  author={Geng, Shijie and Liu, Shuchang and Fu, Zuohui and Ge, Yingqiang and Zhang, Yongfeng},
  booktitle={Proceedings of the 16th ACM conference on recommender systems},
  pages={299--315},
  year={2022}
}

@inproceedings{zhou2025recbase,
  title={Recbase: Generative foundation model pretraining for zero-shot recommendation},
  author={Zhou, Sashuai and Gan, Weinan and Liu, Qijiong and Lei, Ke and Zhu, Jieming and Huang, Hai and Xia, Yan and Tang, Ruiming and Dong, Zhenhua and Zhao, Zhou},
  booktitle={Proceedings of the 2025 Conference on Empirical Methods in Natural Language Processing},
  pages={15598--15610},
  year={2025}
}

@article{10.1145/3716393,
author = {Bao, Keqin and Zhang, Jizhi and Wang, Wenjie and Zhang, Yang and Yang, Zhengyi and Luo, Yanchen and Chen, Chong and Feng, Fuli and Tian, Qi},
title = {A Bi-Step Grounding Paradigm for Large Language Models in Recommendation Systems},
year = {2025},
issue_date = {December 2025},
publisher = {Association for Computing Machinery},
address = {New York, NY, USA},
volume = {3},
number = {4},
url = {https://doi.org/10.1145/3716393},
doi = {10.1145/3716393},
journal = {ACM Trans. Recomm. Syst.},
month = apr,
articleno = {53},
numpages = {27}
}

\clearpage
\appendix
\onecolumn
\section{Appendix}
\label{sec:appendix}
\subsection{Implementation Details}
\label{app:implementation}

All experiments were run on 2$\times$NVIDIA RTX 4090 GPUs (24GB
each). Training used PyTorch. The user projection $P_u$ and item
projection $P_i$ are each a 3-layer MLP with LayerNorm, mapping
their respective input dimensions to the shared projection space. Models are trained for up to 4,000
epochs with early stopping, batch size 16,384, and the Adam
optimizer with a learning rate of $1\times10^{-3}$. Full loss
weights and codebook configuration are reported in
Table~\ref{tab:hyperparams_app}.

\subsection{Theoretical Justification}
\label{app:theory}

This appendix gives the complete definitions, theorem statements,
and proof sketches summarized in the Proposed Model section for
$\mathcal{L}_{\mathrm{gw}}$ and $\mathcal{L}_{\mathrm{adv}}$, and
shows explicitly how each borrowed result connects to the loss
terms used in the main paper's training objective.

\subsubsection{Gromov--Wasserstein Alignment: Full Derivation}
\label{app:theory_gw}

\paragraph{Metric measure spaces.}
We model each source domain's user distribution as a triple
$X_k = (Z_k, d_{\cos}, \mu_k)$, a metric measure space
(mm-space)~\cite{Mmoli2011} (Definition~5.1): $Z_k$ is the set of
user projections $z_u^k$ in domain $k$, $d_{\cos}$ is the cosine
distance, and $\mu_k$ assigns equal weight $\tfrac{1}{|\mathcal{U}^k|}$
to each user.

\paragraph{Continuous GW distance.}
The Gromov--Wasserstein distance between two mm-spaces $X_a$ and
$X_b$ finds the best soft matching $\pi$ between the two user
populations~\cite{Mmoli2011} (Definition~5.7, Eq.~5.9):
\begin{equation}
    \mathfrak{D}_2(X_a, X_b)
    \;=\;
    \inf_{\pi \in \mathcal{M}(\mu_a,\mu_b)}
    \frac{1}{2}
    \left(
    \int\!\!\int_{(Z_a \times Z_b)^2}
    \bigl|d_{\cos}(x,x') - d_{\cos}(y,y')\bigr|^2
    \,d\pi(x,y)\,d\pi(x',y')
    \right)^{\!\!\frac{1}{2}},
    \label{eq:gw_continuous_app}
\end{equation}
where $\mathcal{M}(\mu_a,\mu_b)$ is the set of valid couplings
between $\mu_a$ and $\mu_b$.

\begin{theorem}[M\'emoli~\cite{Mmoli2011}, Theorem~5.1(a)]
\label{thm:gw_isomorphism_app}
$\mathfrak{D}_2(X_a, X_b) = 0$ if and only if $X_a$ and $X_b$ are
isomorphic as mm-spaces, i.e., there exists a measure-preserving
isometry between them.
\end{theorem}

\noindent
This is the target our training loss approaches: driving
$\mathfrak{D}_2 \to 0$ means all source-domain user populations
share the same geometric structure, which an unseen target domain
with similar preference geometry then naturally inherits at
inference.

\paragraph{Why direct minimization fails.}
Finding the optimal $\pi$ in Eq.~\eqref{eq:gw_continuous_app}
requires solving a Quadratic Assignment Problem, NP-hard and
$O(n^3)$ per domain pair~\cite{Mmoli2011}. We require a tractable
proxy.

\paragraph{Distance distribution and the Second Lower Bound.}
For each domain $k$, consider the distribution of cosine distances
between all pairs of users in $X_k$. $\mathrm{SLB}_2$ measures how
different these distributions are across two domains.

\begin{proposition}[M\'emoli~\cite{Mmoli2011}, Proposition~6.2]
\label{prop:slb_app}
\begin{equation}
    \mathrm{SLB}_2(X_a, X_b) \;\le\; \mathfrak{D}_2(X_a, X_b).
    \label{eq:slb_ineq_app}
\end{equation}
\end{proposition}

\noindent
Driving $\mathrm{SLB}_2 \to 0$ is therefore a \emph{necessary} step
toward the isomorphism target of Theorem~\ref{thm:gw_isomorphism_app}.

The quantile function $F_k^{-1}(s)$ of domain $k$'s distance
distribution is the $s$-th percentile of pairwise distances in that
domain. Corollary~6.2 of M\'emoli~\cite{Mmoli2011} gives:

\begin{proposition}[M\'emoli~\cite{Mmoli2011}, Corollary~6.2]
\label{prop:slb_quantile_app}
\begin{equation}
    \bigl(\mathrm{SLB}_2(X_a,X_b)\bigr)^2
    \;\ge\;
    \frac{1}{4}
    \int_0^1
    \bigl(F_a^{-1}(s) - F_b^{-1}(s)\bigr)^2\,ds.
    \label{eq:slb_quantile_app}
\end{equation}
\end{proposition}

\begin{proposition}[Quantile integral to discrete sum;
  P\'eyr\'e \& Cuturi~\cite{peyre2019computational}, Remark~2.28]
\label{prop:discrete_app}
For two empirical distance distributions of equal size $N$:
\begin{equation}
    \int_0^1
    \bigl(F_a^{-1}(s) - F_b^{-1}(s)\bigr)^2\,ds
    \;=\;
    \frac{1}{N}\sum_{r=1}^{N}
    \bigl(d^a_{(r)} - d^b_{(r)}\bigr)^2.
    \label{eq:integral_to_sum_app}
\end{equation}
\end{proposition}

\noindent
Sorting the $N$ pairwise distances and comparing them rank-by-rank
is exact here, not an approximation. The right-hand side is
precisely the $\mathrm{GWLB}^2(a,b)$ term used in the main paper:
\begin{equation}
    \mathrm{GWLB}^2(a,b)
    \;=\;
    W_2^2(\mu^a,\mu^b)
    \;=\;
    \frac{1}{N}\sum_{r=1}^{N}
    \bigl(d^a_{(r)} - d^b_{(r)}\bigr)^2.
    \label{eq:gwlb_main}
\end{equation}

Substituting Eq.~\eqref{eq:integral_to_sum_app} into
Eq.~\eqref{eq:slb_quantile_app} and using Eq.~\eqref{eq:gwlb_main}:
\begin{equation}
    \bigl(\mathrm{SLB}_2(X_a,X_b)\bigr)^2
    \;\ge\;
    \frac{1}{N}\sum_{r=1}^{N}\bigl(d^a_{(r)}-d^b_{(r)}\bigr)^2
    \;=\;
    \,\mathrm{GWLB}^2(a,b).
    \label{eq:slb_to_gwlb_app}
\end{equation}

\begin{corollary}[Theoretical grounding of $\mathcal{L}_{\mathrm{gw}}$]
\label{cor:chain_app}
\begin{equation}
    \mathcal{L}_{\mathrm{gw}} \to 0
    \;\Rightarrow\;
    \mathrm{GWLB}^2 \to 0
    \;\Rightarrow\;
    \mathrm{SLB}_2 \to 0
    \;\xRightarrow{\text{Prop.~\ref{prop:slb_app}}}\;
    \mathfrak{D}_2 \to 0
    \;\xRightarrow{\text{Thm.~\ref{thm:gw_isomorphism_app}}}\;
    X_a \cong X_b.
    \label{eq:chain_app}
\end{equation}
\end{corollary}

\noindent
This chain is the formal justification for the user-unification
loss used in the main paper, which averages $\mathrm{GWLB}^2$
(Eq.~\eqref{eq:gwlb_main}) over all $\binom{K}{2}$ domain pairs:
\begin{equation}
    \mathcal{L}_{\mathrm{gw}}
    \;=\;
    \frac{1}{\binom{K}{2}}
    \sum_{1 \le a < b \le K}
    \mathrm{GWLB}^2(a,b).
    \label{eq:lgw_main}
\end{equation}
By Corollary~\ref{cor:chain_app}, minimizing Eq.~\eqref{eq:lgw_main}
during training is a provably necessary step toward aligning the
internal geometry of every source domain's user population.

\subsubsection{Adversarial Alignment: Full Derivation}
\label{app:theory_adv}

\paragraph{Uniform convergence.}
\begin{lemma}[Ben-David et al.~\cite{ben2010theory}, Lemma~1]
\label{lem:convergence_app}
For $\mathcal{H}$ of VC dimension $d$ and samples $U_S,U_T$ of
size $m$, with probability $\geq 1-\delta$:
\begin{equation}
    d_\mathcal{H}(\mathcal{D}_S,\mathcal{D}_T)
    \leq
    \hat{d}_\mathcal{H}(U_S,U_T)
    + 4\sqrt{\frac{d\log(2m)+\log(2/\delta)}{m}}.
    \label{eq:convergence_app}
\end{equation}
\end{lemma}
\noindent
At batch size $m \approx 3{,}300$ per domain pair, the second term
is negligible, so empirical divergence minimized during training
reliably tracks the population divergence bounded in
Theorem~\ref{thm:bendavid_app}.

\paragraph{Target error bound.}
\begin{theorem}[Ben-David et al.~\cite{ben2010theory}, Theorem~2]
\label{thm:bendavid_app}
For any $h\in\mathcal{H}$ (VC dimension $d$) and unlabeled samples
$U_S,U_T$ of size $m'$, with probability $\geq 1-\delta$:
\begin{equation}
    \epsilon_T(h) \;\leq\;
    \epsilon_S(h)
    + \tfrac{1}{2}\,\hat{d}_{\mathcal{H}\Delta\mathcal{H}}(U_S,U_T)
    + 4\!\sqrt{\tfrac{2d\log 2m' + \log(2/\delta)}{m'}}
    + \lambda,
    \label{eq:bendavid_app}
\end{equation}
where $\lambda{=}\min_h[\epsilon_S(h){+}\epsilon_T(h)]$ is the ideal
joint error.
\end{theorem}

\noindent
Only $\hat{d}_{\mathcal{H}\Delta\mathcal{H}}$ is controllable
through representation learning; this is exactly what the item
discriminator $\psi : \mathbb{S}^{d-1} \to \Delta^{K-1}$ and item
projector $P_i$ target via the minimax game used in the main paper:
\begin{equation}
    \min_{\theta_\psi}\ \max_{\theta_{P_i}}\
    -\frac{1}{K}\sum_{k=1}^{K}\mathbb{E}_{\mathbf{z}\sim P^{k}_{i}}
    \left[\log\psi(\mathbf{z})_{k}\right].
    \label{eq:minimax_main}
\end{equation}

\begin{lemma}[Ben-David et al.~\cite{ben2010theory}, Lemma~2]
\label{lem:pad_app}
For symmetric $\mathcal{H}$ and samples $U_S,U_T$ of size $m$:
\begin{equation}
    \hat{d}_\mathcal{H}(U_S,U_T)
    = 2\Bigl(1-\min_{h\in\mathcal{H}}\Bigl[
        \tfrac{1}{m}\!\sum_{x:h(x)=0}\!\mathbf{I}[x\in U_S]
        +\tfrac{1}{m}\!\sum_{x:h(x)=1}\!\mathbf{I}[x\in U_T]
      \Bigr]\Bigr).
    \label{eq:lemma2_app}
\end{equation}
\end{lemma}
\noindent
This minimum equals the error $\varepsilon$ of the best domain
classifier, giving $\hat{d}_\mathcal{H}=2(1-2\varepsilon)$. At
chance accuracy $\varepsilon \to (K{-}1)/K$, $\hat{d}_\mathcal{H} \to
0$, tightening the bound in Eq.~\eqref{eq:bendavid_app}. This is
precisely the quantity our Proxy $\mathcal{A}$-distance diagnostic
in the main text is designed to estimate.

\paragraph{GAN background and $K$-way adaptation.}
The binary GAN minimax
game~\cite{goodfellow2014generativeadversarialnetworks}:
\begin{equation}
    \min_G\max_D V(D,G)
    =\mathbb{E}_{x\sim p_{\mathrm{data}}}[\log D(x)]
    +\mathbb{E}_{z\sim p_z}[\log(1-D(G(z)))].
    \label{eq:gan_main_app}
\end{equation}
We have no generator and no fake data: $P_i$ projects real SBERT
encodings. We replace the binary real/fake objective with $K$-way
cross-entropy over $K$ real source-domain distributions:
\begin{equation}
    \mathcal{L}_{\mathrm{dom}}(\theta_\psi,\theta_{P_i})
    = -\frac{1}{K}\sum_{k=1}^K
    \mathbb{E}_{\mathbf{z}\sim P^k_i(\theta_{P_i})}
    [\log\psi(\mathbf{z};\theta_\psi)_k].
    \label{eq:disc_app}
\end{equation}
This is exactly Eq.~\eqref{eq:minimax_main} written as a single
loss to be minimized by $\psi$ and maximized by $P_i$.

\begin{proposition}[Goodfellow et al.~\cite{goodfellow2014generativeadversarialnetworks},
  Proposition~1, adapted to $K$ domains]
\label{prop:optimal_disc_app}
Fix $P_i$. The discriminator $\psi^*$ maximizing $\mathcal{L}_{\mathrm{dom}}$
is
\begin{equation}
    \psi^*(\mathbf{z})_k
    = \frac{p_k(\mathbf{z})}{\sum_{j=1}^{K}p_j(\mathbf{z})},
    \quad k=1,\ldots,K,
    \label{eq:optimal_disc_app}
\end{equation}
where $p_k(\mathbf{z})$ is the density of domain-$k$ item
embeddings at $\mathbf{z}$. When all $K$ distributions coincide,
$\psi^*(\mathbf{z})_k=1/K$: the saddle point of the minimax game
in Eq.~\eqref{eq:minimax_main}.
\end{proposition}

\noindent
\textit{Proof sketch.} For fixed $P_i$, $\psi$ maximizes
$\sum_k p_k(\mathbf{z})\log\psi(\mathbf{z})_k$ pointwise over the
simplex; concavity of $\log$ and the Lagrangian first-order
conditions give $\psi^*_k \propto p_k$, yielding
Eq.~\eqref{eq:optimal_disc_app} after normalization --- the direct
$K$-class extension of Proposition~1
in~\cite{goodfellow2014generativeadversarialnetworks}. \hfill$\square$

\noindent
Our single $K$-way discriminator differs from
MADA~\cite{pei2018multi}, which deploys $K$ class-wise binary
discriminators for single source--target adaptation; ours aligns
$K$ source domains simultaneously.

\paragraph{GRL equivalence.}
The minimax game in Eq.~\eqref{eq:minimax_main} naively requires
two alternating gradient passes per step. A gradient reversal layer
(Ganin et al.~\cite{pmlr-v37-ganin15}, Eqs.~4--9) collapses this
into one, since
\begin{equation}
    \frac{\partial\mathcal{L}_{\mathrm{adv}}}{\partial\theta_{P_i}}
    = -\lambda\,
    \frac{\partial\mathcal{L}_{\mathrm{dom}}}{\partial\theta_{P_i}}.
    \label{eq:grad_flip_app}
\end{equation}
This is exactly the mechanism behind the adversarial loss used in
the main paper:
\begin{equation}
    \mathcal{L}_{\mathrm{adv}}
    = -\frac{1}{K}\sum_{k=1}^{K}
    \mathbb{E}_{\mathbf{z}\sim P^{k}_{i}}
    \left[\log\psi\!\left(\mathrm{GRL}_{\lambda}(\mathbf{z})\right)_{k}\right].
    \label{eq:adv_main}
\end{equation}

\begin{proposition}[Ganin et al.~\cite{pmlr-v37-ganin15}, Eqs.~4--9]
\label{prop:grl_equivalence_app}
A descent step on $(\theta_\psi,\theta_{P_i})$ using
$\nabla\mathcal{L}_{\mathrm{adv}}$ implements simultaneously:
(i) descent on $\theta_\psi$; (ii) ascent on $\theta_{P_i}$.
\end{proposition}

\noindent
\textit{Proof.} By the chain rule through $\mathrm{GRL}_\lambda$
(Eq.~\eqref{eq:grad_flip_app}): descent on $\theta_{P_i}$ via
$\mathcal{L}_{\mathrm{adv}}$ becomes ascent on
$\mathcal{L}_{\mathrm{dom}}$. For $\theta_\psi$, GRL is
the identity forward, so
$\partial\mathcal{L}_{\mathrm{adv}}/\partial\theta_\psi
=\partial\mathcal{L}_{\mathrm{dom}}/\partial\theta_\psi$. Both
directions execute in one backward pass, exactly as implemented in
Eq.~\eqref{eq:adv_main}. \hfill$\square$

\paragraph{Scope and novelty.}
These theorems were originally derived for classification
(Ben-David et al.), generative modeling (Goodfellow et al.), and
single source--target visual adaptation (Ganin et al.). We adapt
them to multi-source zero-shot recommendation, where the downstream
task is ranking and alignment operates on projected item
embeddings rather than raw input features. To our knowledge, this
is the first application of the Ben-David divergence framework to
item-space alignment in recommendation. 

\FloatBarrier
\subsection{Modality-Dropout Ablation}
\label{app:modality_dropout}

To assess whether the collaborative-to-zero distribution shift at
zero-shot inference affects performance, we train a variant of
\textsc{Atlas} with modality dropout: during training, the
collaborative component $e_{u,c}^k$ is replaced with a zero vector
with probability $p_{\mathrm{drop}}=0.2$ before concatenation,
exposing $P_u$ to the same zero-imputed input distribution used at
zero-shot inference. Table~\ref{tab:modality_dropout_app} compares
zero-shot performance with and without modality dropout, and
Table~\ref{tab:modality_dropout_source_app} compares performance on
the five source domains.

\begin{table}[!htbp]
\centering
\small
\caption{Zero-shot performance with and without modality dropout on target domains.}
\label{tab:modality_dropout_app}
\begin{tabular}{l@{\hspace{4ex}}cc@{\hspace{4ex}}cc}
\toprule
& \multicolumn{2}{c}{No Dropout} & \multicolumn{2}{c}{With Dropout} \\
\cmidrule(lr){2-3} \cmidrule(lr){4-5}
\textbf{Target Domain} & HR & NDCG & HR & NDCG \\
\midrule
Digital Music        & .0173 & .0089 & .0191 & .0113 \\
Cell Phones          & .0023 & .0012 & .0027 & .0013 \\
Arts \& Crafts       & .0067 & .0035 & .0069 & .0037 \\
Baby Products        & .0042 & .0022 & .0042 & .0022 \\
Books                & .0055 & .0027 & .0063 & .0031 \\
Musical Instruments  & .0061 & .0038 & .0057 & .0035 \\
Office Products      & .0026 & .0012 & .0031 & .0015 \\
Appliances           & .0113 & .0059 & .0103 & .0053 \\
Software             & .0073 & .0043 & .0074 & .0043 \\
Pet Supplies         & .0009 & .0003 & .0006 & .0002 \\
\bottomrule
\end{tabular}
\end{table}

\begin{table}[!htbp]
\centering
\small
\caption{Source-domain performance with and without modality dropout.}
\label{tab:modality_dropout_source_app}
\begin{tabular}{l@{\hspace{4ex}}cc@{\hspace{4ex}}cc}
\toprule
& \multicolumn{2}{c}{No Dropout} & \multicolumn{2}{c}{With Dropout} \\
\cmidrule(lr){2-3} \cmidrule(lr){4-5}
\textbf{Source Domain} & HR & NDCG & HR & NDCG \\
\midrule
Beauty        & .0870 & .0490 & .0692 & .0436 \\
Automotive    & .0532 & .0332 & .0441 & .0303 \\
Movies \& TV  & .1388 & .0774 & .1181 & .0701 \\
Video Games   & .2116 & .1561 & .1777 & .1250 \\
Electronics   & .0687 & .0443 & .0552 & .0394 \\
\bottomrule
\end{tabular}
\end{table}

\noindent
Modality dropout yields a marginal improvement in zero-shot
performance, consistent with the hypothesis that exposing $P_u$ to
the zero-imputed input setting during training slightly narrows the
train--inference distribution gap. However, this comes at a
noticeable cost to source-domain performance, since randomly
ablating the collaborative signal during training reduces the
effective training signal available for in-domain recommendation.
Given this trade-off, the gain in the zero-shot setting is
small relative to the loss on source domains.

\FloatBarrier
\subsection{RDG Setting Overview}
\label{app:rdg_overview}

Figure~\ref{fig:rdg_overview_app} contrasts Recommendation Domain
Generalization with existing recommendation transfer paradigms
(domain-specific, cross-domain, and universal/foundation-model
recommendation), highlighting that RDG is the only setting in which
a single frozen model, trained without any target-domain data, is
deployed directly to multiple completely unseen recommendation
domains.

\begin{figure}[!htbp]
\centering
\includegraphics[width=\linewidth]{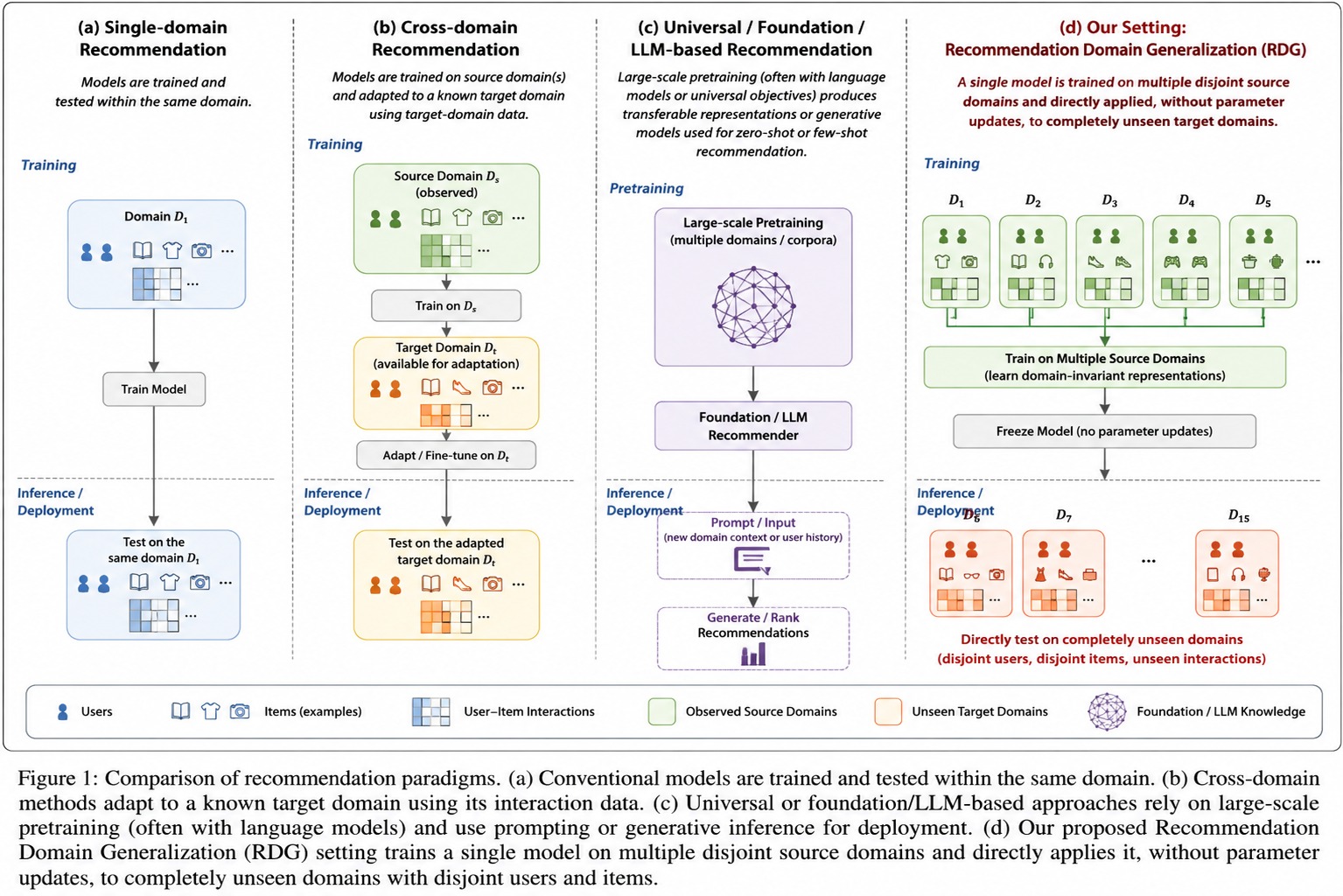}
\caption{}
\label{fig:rdg_overview_app}
\end{figure}

\FloatBarrier
\subsection{Full Experimental Results}
\label{app:full_results}

\subsubsection{Dataset Statistics}
\label{app:datasets}

\begin{table}[!htbp]
\centering
\small
\caption{Dataset statistics (Amazon Reviews 2023, filtered to
$\geq 5$ interactions per user).}
\label{tab:datasets_app}
\begin{tabular}{l@{\hspace{4ex}}rrr}
\toprule
\textbf{Domain} & \textbf{\#Users} & \textbf{\#Items} & \textbf{Avg. Ints.} \\
\midrule
\multicolumn{4}{l}{\textit{Source Domains (Training)}} \\
\midrule
Beauty \& Personal Care & 588,804 & 104,891 & 54.60 \\
Automotive              & 489,076 & 160,114 & 41.20 \\
Movies \& TV            & 726,286 & 117,180 & 76.80 \\
Video Games             & 117,742 &  43,023 & 17.90 \\
Electronics             & 881,540 & 114,302 & 63.20 \\
\midrule
\textbf{Total} & \textbf{2,803,448} & \textbf{539,510} & -- \\
\midrule
\multicolumn{4}{l}{\textit{Target Domains (Zero-Shot Evaluation)}} \\
\midrule
Digital Music        &   1,465 &  13,655 & 10.70 \\
Cell Phones          & 308,090 & 128,829 & 22.90 \\
Arts \& Crafts        & 302,491 &  98,259 & 30.70 \\
Baby Products         & 184,851 &  87,474 & 19.60 \\
Books                 & 237,965 & 298,412 & 48.30 \\
Musical Instruments   &  84,574 &  89,523 & 16.90 \\
Office Products       & 332,744 & 153,821 & 24.20 \\
Appliances            &   9,918 &  20,503 &  6.40 \\
Software              & 160,256 &  41,200 & 20.70 \\
Pet Supplies          & 331,119 & 143,272 & 33.60 \\
\midrule
\textbf{Total} & \textbf{1,953,473} & \textbf{1,074,948} & -- \\
\bottomrule
\end{tabular}
\end{table}

\FloatBarrier
\subsubsection{Codebook vs.\ Projected Embeddings: Full Results}
\label{app:zrec_zproj}

\begin{table}[!htbp]
\centering
\small
\caption{HR@10 comparison of continuous projection $z_{\mathrm{proj}}$
vs.\ codebook-reconstructed $z_{\mathrm{rec}}$ across all ten unseen
target domains.}
\label{tab:zrec_vs_zproj_app}
\begin{tabular}{l@{\hspace{4ex}}cc}
\toprule
\textbf{Target Domain} & $z_{\mathrm{proj}}$ & $z_{\mathrm{rec}}$ \\
\midrule
Digital Music         & 0.0155 & \textbf{0.0173} \\
Cell Phones           & 0.0009 & \textbf{0.0023} \\
Arts \& Crafts         & 0.0039 & \textbf{0.0067} \\
Baby Products          & 0.0023 & \textbf{0.0042} \\
Books                  & 0.0012 & \textbf{0.0055} \\
Musical Instruments    & 0.0024 & \textbf{0.0061} \\
Office Products        & 0.0017 & \textbf{0.0026} \\
Appliances             & 0.0085 & \textbf{0.0113} \\
Software               & 0.0043 & \textbf{0.0073} \\
Pet Supplies           & 0.0006 & \textbf{0.0009} \\
\midrule
\textbf{Average} & 0.0041 & \textbf{0.0064} \\
\bottomrule
\end{tabular}
\end{table}

\FloatBarrier
\subsubsection{Ablation on Alignment Objectives: Full Results}
\label{app:ablation_full}

\begin{table}[htbp]
\centering
\small
\caption{Ablation on alignment objectives (HR@10) across seen (source) and unseen (target) domains.}
\label{tab:ablation_full_app}
\setlength{\tabcolsep}{3pt} 
\begin{tabular}{l ccccc c cccccccccc}
\toprule
& \multicolumn{5}{c}{\textit{Seen (Source)}} & \multicolumn{10}{c}{\textit{Unseen (Target)}} \\
\cmidrule(lr){2-6} \cmidrule(lr){7-16}
\textbf{Configuration} & Mov. & Gam. & Auto. & Elec. & Beau. & Cell & Mus. & Arts & Baby & Books & Inst. & Off. & Appl. & Soft. & Pet \\
\midrule
BPR only              & .0332 & .0941 & .0074 & .0088 & .0205 & .0004 & .0061 & .0018 & .0014 & .0020 & .0029 & .0011 & .0041 & .0021 & .0004 \\
BPR + $\mathcal{L}_\mathrm{adv}$ & .0711 & .1421 & .0396 & .0344 & .0761 & .0002 & .0072 & .0053 & .0009 & .0042 & .0044 & .0020 & .0066 & .0061 & .0007 \\
BPR + $\mathcal{L}_\mathrm{GW}$  & .0420 & .1893 & .0203 & .0442 & .0491 & .0009 & .0107 & .0021 & .0032 & .0033 & .0041 & .0015 & .0091 & .0036 & .0004 \\
\midrule
\textsc{Atlas} (full) & \textbf{.1388} & \textbf{.2116} & \textbf{.0532} & \textbf{.0687} & \textbf{.0870} & \textbf{.0023} & \textbf{.0173} & \textbf{.0067} & \textbf{.0042} & \textbf{.0055} & \textbf{.0061} & \textbf{.0026} & \textbf{.0113} & \textbf{.0073} & \textbf{.0009} \\
\bottomrule
\end{tabular}
\end{table}

\FloatBarrier
\subsubsection{Source-Domain Diversity Effect: Full Results}
\label{app:diversity_full}

\begin{table}[htbp]
\centering
\small
\caption{Source-domain diversity effect (HR@10) as $K$ increases from 1 to 5. ($M$=Movies, $E$=Electronics, $B$=Beauty, $V$=Video Games, $A$=Automotive).}
\label{tab:diversity_full_app}
\setlength{\tabcolsep}{2.5pt}
\begin{tabular}{cl ccccc c cccccccccc}
\toprule
& & \multicolumn{5}{c}{\textit{Seen (Source)}} & \multicolumn{10}{c}{\textit{Unseen (Target)}} \\
\cmidrule(lr){3-7} \cmidrule(lr){8-17}
$K$ & Sources & Mov. & Elec. & Beau. & Gam. & Auto. & Mus. & Cell & Arts & Baby & Books & Inst. & Off. & Appl. & Soft. & Pet \\
\midrule
1 & M       & .2359 & --    & --    & --    & --    & .0028 & .0001 & .0001 & .0001 & .0002 & .0001 & .0001 & .0009 & .0006 & .0001 \\
2 & ME      & .0103 & .0180 & --    & --    & --    & .0071 & .0005 & .0011 & .0005 & .0004 & .0008 & .0001 & .0027 & .0013 & .0004 \\
3 & MEB     & .0286 & .0251 & .0296 & --    & --    & .0112 & .0009 & .0020 & .0014 & .0016 & .0017 & .0007 & .0057 & .0028 & .0002 \\
4 & MEBV    & .0575 & .0399 & .0434 & .1883 & --    & .0136 & .0015 & .0040 & .0034 & .0038 & .0041 & .0014 & .0072 & .0056 & .0007 \\
5 & MEBVA   & \textbf{.1388} & \textbf{.0687} & \textbf{.0870} & \textbf{.2116} & \textbf{.0532} & \textbf{.0173} & \textbf{.0023} & \textbf{.0067} & \textbf{.0042} & \textbf{.0055} & \textbf{.0061} & \textbf{.0026} & \textbf{.0113} & \textbf{.0073} & \textbf{.0009} \\
\bottomrule
\end{tabular}
\end{table}

\FloatBarrier
\subsubsection{Sensitivity to Random Initialization}
\label{app:seeds}

To assess the stability of our headline results, we train
\textsc{Atlas}$_\text{ZS}$ across three random seeds (42, 37, 61;
seed 42 is the main paper configuration) and report per-seed values
together with the mean $\pm$ standard deviation per domain. Standard
deviations are consistently small relative to the mean, indicating
that \textsc{Atlas}'s reported gains are stable across
initializations. Baseline results follow their original single-run
reporting protocols, consistent with prior work in this setting.

\begin{table}[htbp]
\centering
\small
\caption{\textsc{Atlas}$_\text{ZS}$ HR@10 / NDCG@10 across random seeds, with mean $\pm$ std.}
\label{tab:seeds_zs_raw_app}
\setlength{\tabcolsep}{5pt}
\begin{tabular}{l cc cc cc cc}
\toprule
& \multicolumn{2}{c}{Seed 42} & \multicolumn{2}{c}{Seed 37} & \multicolumn{2}{c}{Seed 61} & \multicolumn{2}{c}{Mean $\pm$ Std} \\
\cmidrule(lr){2-3} \cmidrule(lr){4-5} \cmidrule(lr){6-7} \cmidrule(lr){8-9}
\textbf{Target Domain} & HR & NDCG & HR & NDCG & HR & NDCG & HR & NDCG \\
\midrule
Digital Music       & .0173 & .0089 & .0181 & .0092 & .0159 & .0083 & $.0171\pm.0009$ & $.0088\pm.0004$ \\
Cell Phones         & .0023 & .0012 & .0019 & .0012 & .0022 & .0012 & $.0021\pm.0002$ & $.0012\pm.0000$ \\
Arts \& Crafts      & .0067 & .0035 & .0063 & .0033 & .0061 & .0032 & $.0064\pm.0002$ & $.0033\pm.0001$ \\
Baby Products       & .0042 & .0022 & .0045 & .0024 & .0040 & .0021 & $.0042\pm.0002$ & $.0022\pm.0001$ \\
Books               & .0055 & .0027 & .0059 & .0029 & .0053 & .0025 & $.0056\pm.0002$ & $.0027\pm.0002$ \\
Musical Instruments & .0061 & .0038 & .0064 & .0040 & .0059 & .0035 & $.0061\pm.0002$ & $.0038\pm.0002$ \\
Office Products     & .0026 & .0012 & .0029 & .0015 & .0025 & .0011 & $.0027\pm.0002$ & $.0013\pm.0002$ \\
Appliances          & .0113 & .0059 & .0101 & .0055 & .0115 & .0057 & $.0110\pm.0006$ & $.0057\pm.0002$ \\
Software            & .0073 & .0043 & .0077 & .0045 & .0072 & .0042 & $.0074\pm.0002$ & $.0043\pm.0001$ \\
Pet Supplies        & .0009 & .0003 & .0009 & .0003 & .0008 & .0003 & $.0009\pm.0000$ & $.0003\pm.0000$ \\
\bottomrule
\end{tabular}
\end{table}

\begin{table*}[!htbp]
\centering
\small
\caption{\textsc{Atlas} in-domain HR@10 / NDCG@10 for each individual seed, with mean $\pm$ std across all three.}
\label{tab:seeds_source_raw_app}
\resizebox{\textwidth}{!}{
\begin{tabular}{l@{\hspace{1.5em}}cc@{\hspace{1.5em}}cc@{\hspace{1.5em}}cc@{\hspace{2em}}cc}
\toprule
& \multicolumn{2}{c}{Seed 42} & \multicolumn{2}{c}{Seed 37} & \multicolumn{2}{c}{Seed 61} & \multicolumn{2}{c}{Mean $\pm$ Std} \\
\cmidrule(lr){2-3} \cmidrule(lr){4-5} \cmidrule(lr){6-7} \cmidrule(lr){8-9}
\textbf{Source Domain} & HR & NDCG & HR & NDCG & HR & NDCG & HR & NDCG \\
\midrule
Beauty         & .0870 & .0490 & .0895 & .0497 & .0873 & .0490 & $.0879\pm.0011$ & $.0492\pm.0003$ \\
Automotive     & .0532 & .0332 & .0497 & .0318 & .0527 & .0328 & $.0519\pm.0015$ & $.0326\pm.0006$ \\
Movies \& TV   & .1388 & .0774 & .1504 & .0811 & .1401 & .0787 & $.1431\pm.0052$ & $.0791\pm.0015$ \\
Video Games    & .2116 & .1561 & .1974 & .1394 & .1941 & .1366 & $.2010\pm.0076$ & $.1440\pm.0086$ \\
Electronics    & .0687 & .0443 & .0814 & .0572 & .0618 & .0404 & $.0706\pm.0081$ & $.0473\pm.0072$ \\
\bottomrule
\end{tabular}
}
\end{table*}

\FloatBarrier
\subsection{Hyperparameters}
\label{app:hyperparams}

\begin{table}[!htbp]
\centering
\small
\caption{Architecture and training hyperparameters.}
\label{tab:hyperparams_app}
\begin{tabular}{l@{\hspace{4ex}}l}
\toprule
\textbf{Component} & \textbf{Value} \\
\midrule
User input dim.\ ($e_u^k$)     & 768 (384 CF + 384 SBERT) \\
Item input dim.\ ($e_i^k$)     & 384 (SBERT only) \\
Projection dim.\    & 384 \\
RVQ levels ($L$)               & 3 \\
User codebook sizes  & 512, 1024, 2048 \\
Item codebook sizes  & 512, 1024, 2048 \\
\midrule
Epochs                          & 4000  \\
Batch size                      & 16{,}384 \\
Optimizer                       & Adam \\
Learning rate                   & $1\times10^{-3}$ \\
Warmup epochs                   & 30 \\
Random seed (main)              & 42 \\
\midrule
Softmax temperature ($\tau_{\mathrm{rec}}$) & 0.20 \\
Sinkhorn iterations             & 5 \\
Sinkhorn entropy reg.\          & 0.05 \\
VQ temperature  &  0.25 \\
\midrule
$\lambda_{\mathrm{bpr}}$         & 1.7 \\
$\lambda_{\mathrm{commit}}$           & 0.25 \\
$\lambda_{\mathrm{book}}$             & 2.0 \\
$\lambda_{\mathrm{ent}}$              & 0.05 \\
$\lambda_{\mathrm{gw}}$               & 2.0 \\
$\lambda_{\mathrm{adv}}$  & 2.0 \\
$\lambda_{\mathrm{centroid}}$         & 2.0 \\
$\lambda_{\mathrm{user\_div}}$        & 1.0 \\
$\lambda_{\mathrm{var}}$              & 1.0 \\
\bottomrule
\end{tabular}
\end{table}

\FloatBarrier
\subsection{Representational Diagnostics}
\label{app:repr_diagnostics}

\textbf{Note:} Extended visualizations of these diagnostics
across all training checkpoints are provided in the supplementary
media file accompanying this submission.

\subsubsection{Codebook Utilization}
\label{app:codebook_util}

Figure~\ref{fig:codebook_hist_app} shows the distribution of
codeword assignment frequencies for the user and item codebooks at
the final training checkpoint. Both codebooks exhibit healthy
utilization: assignment mass is spread across the majority of
available codewords at each level rather than concentrating on a
small subset, indicating that the entropy penalty
$\mathcal{L}_{\mathrm{ent}}$ successfully prevents codebook collapse
during training.

\begin{figure}[!htbp]
\centering
\includegraphics[width=\linewidth]{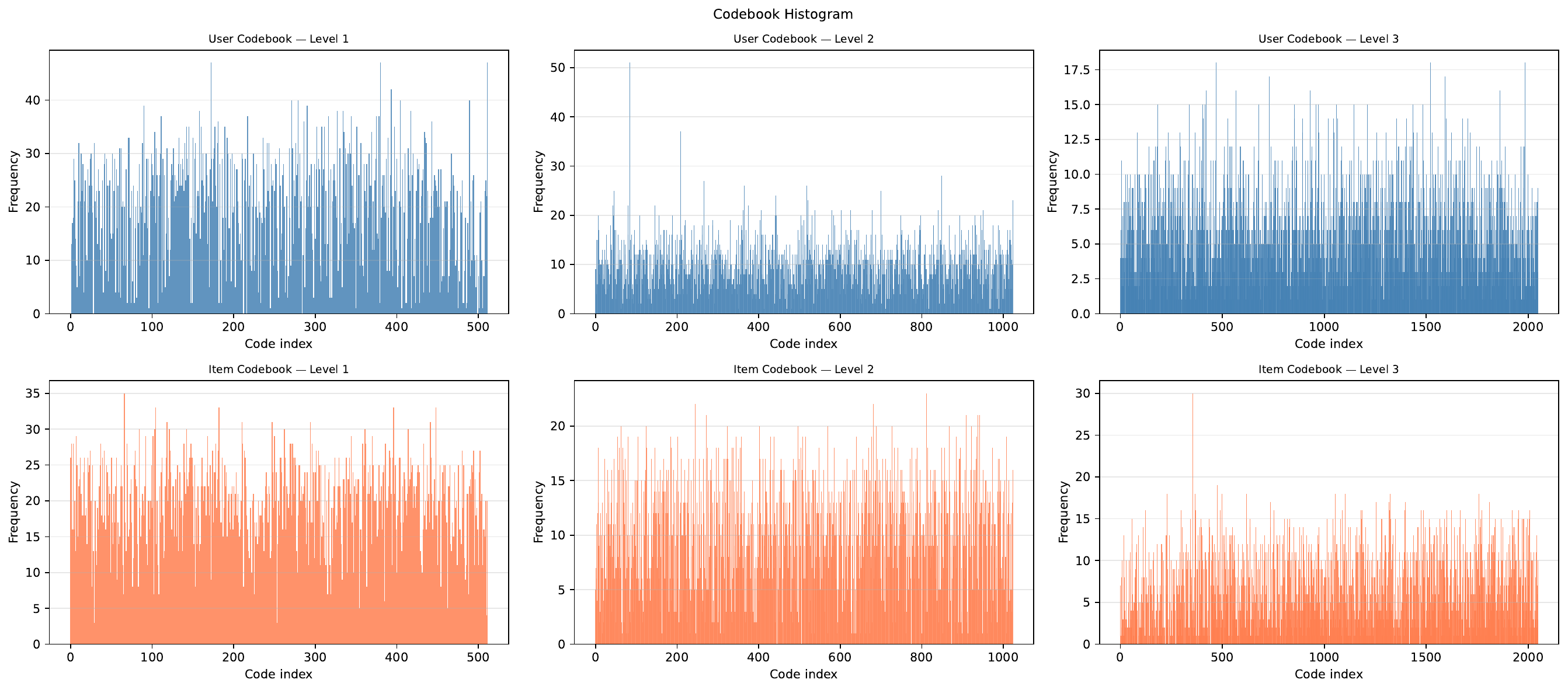}
\caption{Codebook assignment frequency histograms for user (top) and
item (bottom) codebooks, across levels $L_1$, $L_2$, $L_3$.}
\label{fig:codebook_hist_app}
\end{figure}

\FloatBarrier
\subsubsection{Item-Domain Invariance Across Training}
\label{app:pad_epochs}

Figure~\ref{fig:pad_epochs_app} extends Figure~3 of the main paper
by showing the pairwise Proxy $\mathcal{A}$-distance between
source-domain item representations at six training checkpoints
(epochs 1, 100, 500, 1000, 2000, 4000).

\begin{figure}[!htbp]
\centering
\includegraphics[width=\linewidth]{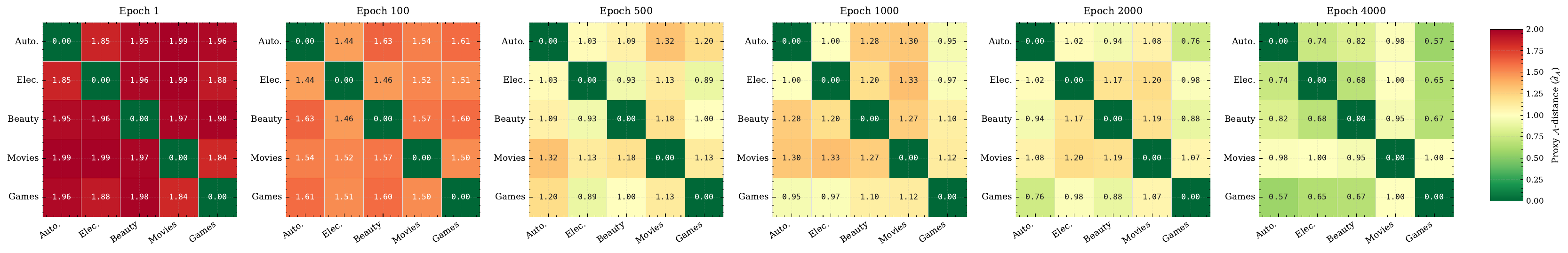}
\caption{Proxy $\mathcal{A}$-distance between source-domain item
representations at epochs 1, 100, 500, 1000, 2000, and 4000.}
\label{fig:pad_epochs_app}
\end{figure}

\end{document}